%% file: paper.tex
\documentclass[10pt]{article}

\usepackage{geometry}  
\usepackage[english]{babel}
\usepackage[utf8x]{inputenc}
\usepackage[T1]{fontenc}

\usepackage[dvipsnames]{xcolor}
\usepackage{paralist}
\usepackage{graphicx}
\usepackage{subcaption}
\usepackage{longtable} 
\usepackage{multirow}
\usepackage{listings}
\usepackage{makecell}
\usepackage{array}
\usepackage{float}
\usepackage{dsfont}
\usepackage{rotating}
\usepackage{booktabs}
\usepackage{enumerate}
\usepackage{tikz}
\usepackage{pgf}

\usepackage{amsmath}
\usepackage{amssymb}
\usepackage{amsthm}
\usepackage{bm}
\usepackage{mathtools}
\usepackage{mathrsfs}
\usepackage{algorithm}
\usepackage{algorithmic}

\usepackage{natbib}
\usepackage[
  colorlinks,
  citecolor=blue,
  linkcolor=red,
  anchorcolor=red,
  urlcolor=blue
]{hyperref}
\mathtoolsset{showonlyrefs}
\usepackage{authblk}
\usepackage{todonotes}
\usepackage{hyperref}
\input{macro}

\usepackage{ying}

\long\def\comment#1{}

\let\hat\widehat
\let\tilde\widetilde

\def \iid {\stackrel{\textnormal{i.i.d.}}{\sim}}

\def \obs {\textnormal{obs}}
\def \TPP {\tilde{\PP}}

\def \att {{\textnormal{ATT}}}
\def \atc {{\textnormal{ATC}}}
\def \ate {{\textnormal{ATE}}}
\def \ORR {{\textnormal{OR}}}
\def \ess {{\textnormal{ess}}}
\def \uly {\underline{y}}
\def \pto {{\stackrel{P}{\to}}}

\usepackage{hyperref}
\def\##1\#{\begin{align}#1\end{align}}
\def\$#1\${\begin{align*}#1\end{align*}}

\makeatletter
\newcommand{\printfnsymbol}[1]{%
  \textsuperscript{\@fnsymbol{#1}}%
}
\makeatother

\usepackage{xcolor}

\title{Sensitivity analysis under the $f$-sensitivity models:\\
a distributional robustness perspective}
\author[1]{Ying Jin\thanks{Author names listed alphabetically.}}
\author[2]{Zhimei Ren\printfnsymbol{1}}
\author[3]{Zhengyuan Zhou}
\affil[1]{Department of Statistics, Stanford University}
\affil[2]{Department of Statistics, University of Chicago}
\affil[3]{Stern School of Business, New York University}
\date{}

\begin{document}

\maketitle

\begin{abstract}

This paper introduces the $f$-sensitivity model, a new 
sensitivity model
that characterizes the violation of unconfoundedness in causal inference.
It assumes the selection bias due to unmeasured confounding
is bounded ``on average''; compared with
the widely used point-wise sensitivity models in the literature, 
it is able to capture the strength of unmeasured confounding 
by not only its magnitude but also the chance of 
encountering such a magnitude. 

We propose a framework for sensitivity analysis
under our new model based on a distributional robustness 
perspective. We first  show that the bounds on counterfactual means
under the $f$-sensitivity model are optimal solutions
to a new class of distributionally robust
optimization (DRO) programs, whose dual forms
are essentially risk minimization problems.
We then construct point estimators for these bounds 
by applying a novel
debiasing technique to 
the output of the corresponding empirical
risk minimization (ERM) problems.  
Our estimators are shown to 
converge to valid bounds on counterfactual means 
if any nuisance component can be 
estimated consistently, and to 
the exact bounds when the ERM step
is additionally consistent. 
We further establish asymptotic normality
and Wald-type inference
for these estimators
under slower-than-root-$n$ convergence 
rates of the estimated nuisance components.
Finally, the performance of
our method is demonstrated with numerical experiments.
\end{abstract}
 
\section{Introduction}
\input{intro}

\section{The new $f$-sensitivity model}
\label{sec:shift}

\input{model}

\section{Sensitivity analysis under the new model}
\label{sec:method}

\input{sec3}

\section{Theoretical guarantees}
\input{sec4}

\section{Numerical experiments}
\input{sec5}

\section{Discussion}
\input{discussion}

\section{Acknowledgement}
The authors thank Emmanuel Cand\`es, Kevin Guo and Dominik Rothenh\"ausler 
for helpful discussions. Z. R. was partially supported by ONR grant 
N00014-20-1-2337, and NIH grants R56HG010812, R01MH113078 and R01MH123157.


\bibliographystyle{apalike}
\bibliography{ref}

\newpage 
\appendix 
\input{appendix}

\end{document}

%% file: macro.tex
\theoremstyle{plain}
\newtheorem{definition}{Definition}
\newtheorem{theorem}{Theorem}
\newtheorem{lemma}{Lemma}

\newtheorem{property}{Property}
\newtheorem{example}{Example}
\newtheorem{corollary}{Corollary}
\newtheorem{remark}{Remark}
\newtheorem{assumption}{Assumption}
\newtheorem{proposition}{Proposition}

\newcommand{\diff}{\textnormal{d}}
\newcommand{\R}{\mathbb{R}}
\newcommand{\E}{\mathbb{E}}
\newcommand{\PP}{\mathbb{P}}
\newcommand{\ind}{\mathds{1}}

\newcommand{\indep}{\rotatebox[origin=c]{90}{$\models$}}

\newcommand{\argmax}[1]{\underset{#1}{\arg\!\max}}
\newcommand{\argmin}[1]{\underset{#1}{\arg\!\min}}

\usepackage{xspace}

\newcommand{\given}{{\,|\,}}
\newcommand{\biggiven}{\,\big{|}\,}
\newcommand{\Biggiven}{\,\Big{|}\,}
\newcommand{\bigggiven}{\,\bigg{|}\,}
\newcommand{\Bigggiven}{\,\Bigg{|}\,}
\def\##1\#{\begin{align}#1\end{align}}
\def\$#1\${\begin{align*}#1\end{align*}}

\definecolor{myblue}{rgb}{.8, .8, 1}
\definecolor{mathblue}{rgb}{0.2472, 0.24, 0.6} 
\definecolor{mathred}{rgb}{0.6, 0.24, 0.442893}
\definecolor{mathyellow}{rgb}{0.6, 0.547014, 0.24}

\newcommand{\calL}{{\mathcal{L}}}

\newcommand{\calU}{{\mathcal{U}}}

\newcommand{\calX}{{\mathcal{X}}}

%% file: intro.tex

In a variety of areas, conducting randomized trial 
can be costly, unethical or even infeasible. 
To draw causal conclusions, researchers/policy-makers 
need to resort to observational data.
The particular challenge in 
observational studies is confounding: 
because the treatment allocation 
mechanism 
is completely unknown, 
there might exist variables that
affect both the treatment and the outcomes. 
With unmeasured confounding, causal 
conclusions drawn from na\"ively comparing the 
outcomes for 
the treated and untreated units -- even after 
adjusting for the difference in 
the observable characteristics -- can be invalid. 

An example is the well-known debate over 
the effect of smoking (the treatment) 
on the development of lung cancer (the outcome), 
where one observes a higher prevalence of
lung cancer among smokers and concludes
that smoking causes lung cancer. 
The criticism of~\citet{fisher1958cigarettes} 
argues that this effect may instead be 
entirely driven by genetics: 
even for two people with the same 
observed characteristics (e.g., demographic
information and medical history), 
the one who is genetically more likely 
to develop lung cancer 
may also be genetically more likely to smoke; 
should this be true, even though smoking may not cause the lung 
cancer, we could still observe a higher proportion 
of lung cancer among the treated group 
even after matching the observed characteristics. 
In this case, the genetic factor 
could be an unmeasured confounder that induces
the nontrivial observed effect and
potentially leads to 
a faulty causal conclusion.

To formalize the discussion,  
we follow the potential outcome framework \citep{neyman1923applications,imbens2015causal}
and posit a data-generating distribution $\PP$ on $(X, U, T, Y(1), Y(0))$, where $X \in \mathcal{X}$ is the observed covariate vector 
in a compact set $\mathcal{X}\subset\RR^p$,
$U \in \mathcal{U}$ is the unobserved confounding factor, 
$T \in \{0,1\}$ is the treatment option ($T=1$ for receiving the treatment 
and $T=0$ for control),
and $Y(1) \in \RR$ and $Y(0) \in \RR$ are the two potential outcomes. 
We assume access to a dataset $\{(X_i, T_i, Y_i)\}_{i=1}^n$ of $n$ i.i.d.~triplets generated from $\PP$, 
where for unit $i$,  
$Y_i = Y_i (T_i)$ is the observed outcome under 
treatment $T_i$,\footnote{here we implicitly make the Stable Unit Treatment Value Assumption (SUTVA).}   
Without loss of generality (since $U$ is arbitrary), under $\PP$, one has 
\#\label{eq:u_confoundedness}
\big(Y(1),Y(0)\big)~\indep~T \given X,U.
\# 

We are interested in 
the average treatment effect (ATE): 
	$\EE[Y(1) - Y(0)]$, 
the average treatment effect on the control (ATC):
	$\EE[Y(1) - Y(0) \given T = 0]$, 
and the average treatment effect on the treated (ATT):
	$\EE[Y(1) - Y(0) \given T = 1]$, 
where in all quantities the expectation is taken with respect to the underlying joint distribution. To make progress, we assume that there is sufficient exploration in the dataset, known as the overlap assumption in the literature. 
We define the observed propensity score $e(x) := \PP(T=1\given X=x)$. 

\begin{assumption}[Overlap]
	 $0 < e(x) < 1$ for $\PP$-almost all $x\in \cX$.\footnote{Since $\mathcal{X}$ is compact, this is equivalent to assuming that $\eta \le e(x) \le 1 -\eta$ for some positive $\eta$, as used in certain versions of overlap in the literature.} 
	\end{assumption}

Under the overlap assumption, the 
identification and estimation of 
treatment effects in observational studies 
have mostly relied on  
the unconfoundedness condition (a.k.a. strong ignorability~\citep{rosenbaum1983central}): 
$ \big(Y(1),Y(0)\big)~\indep~T \given X$. 
That is, all confounders 
that could simultaneously affect the treatment assignment
and the outcomes have been measured in $X$.  
In the lung cancer example, this condition 
imposes that for all people with the same value of $X$, 
even though their genetics and potential outcomes 
differ, they are equally likely to become
a smoker (receive the treatment).

The strong ignorability condition, however,
is not testable and is often hard to justify 
in practice. 
Sensitivity analysis offers a way 
to bypass this obstacle.  
In 
the lung cancer example,~\citet{cornfield1959smoking}  
for the first time used the method of sensitivity analysis: 
it strongly supports the existence of treatment effects 
by showing
that a genetic factor must be nine times more prevalent in smokers than in
non-smokers in order to explain the observed
effect should there be no actual treatment effects
(and it is high implausible to find such a genetic
factor). 
At a high level, sensitivity analysis 
starts with a sensitivity model 
on how the unknown  
data generating process deviates from the
strong ignorability condition, 
and then estimates the range -- rather than a single value -- of the treatment effects, thus 
offering a quantitative understanding 
of how robust the causal conclusion is 
against potential unmeasured confounding.  
The method in~\citet{cornfield1959smoking} was generalized by 
Rosenbaum's $\Gamma$-selection condition~\citep{rosenbaum1987sensitivity}, a pioneering model on the selection bias that has become a classic.~\citet{tan2006distributional}
later proposed the marginal sensitivity model, 
based on which
a series of work have developed various treatment effects estimation and inference schemes 
~\citep{zhao2017sensitivity,kallus2019interval,lee2020causal,dorn2021sharp,jin2021sensitivity,nie2021covariate,dorn2021doubly}. 
The marginal sensitivity model centers around 
the key quantity 
\#\label{eq:def_or}
\ORR(x,u)=\frac{\PP(T = 1 \given X = x) / \PP(T = 0 \given X = x)}{\PP(T = 1 \given X = x, U = u)/
	\PP (T = 0 \given X = x, U = u)},
\#
the odds ratio of receiving treatment conditional 
only on observed covariates versus conditional 
on both unmeasured confounders and observed covariates. 
Intuitively, $\ORR(x,u)$ quantifies 
the impact of unmeasured confounding 
on the treatment probability. 
\citet{tan2006distributional} 
assumes uniformly bounded odds ratio: 
\begin{equation}\label{eq:Tan2006}
 \frac{1}{\Gamma} \le \ORR(x,u)
\le \Gamma  
\end{equation}
for $\PP$-almost all $x \in \calX$ and $u\in \calU$ for some $\Gamma \ge 1$. 
When $\Gamma = 1$, this 
assumption recovers the unconfoundedness assumption, and the larger the $\Gamma$, the more confoundness the model tolerates.
 
Despite being widely used,  the marginal $\Gamma$-selection model~\eqref{eq:Tan2006} can be limited 
in some cases.  
We illustrate this point 
with a simple and natural parametric example, 
which also motivates our new sensitivity model.   

\begin{example}\normalfont~\label{ex:motivation}
Let us consider a simple example  
without covariates. 
We assume the observed probability of treatment is 
$\PP(T=1)=1/2$. 
In this context, the strong ignorability condition means
all units receive treatments with the 
same probability. The researcher would like to 
estimate the range of $\EE[Y(1)\given T=0]$ if the observational 
data is confounded to some extent.   
Based on some background knowledge, 
she is in particular worried about a confounder 
$U\sim \mathcal{N}(0,1)$, where
$T\given U \sim \textsf{Bern}\big(\frac{\exp(\delta U)}{1+\exp(\delta U)}\big)$ 
for some $\delta\in(0,1)$.\footnote{
Here, $Y(1) = g(U,\delta)$ for some 
measurable function $g$, so 
that~\eqref{eq:u_confoundedness} holds. 
One can show that for any $\delta$, 
the distribution of $T$ and $Y(1)$ agrees with 
the observable if $g$ is properly chosen. 
} 
By construction, the odds ratio characterizing 
the selection bias caused by $U$ is
\#\label{ex:logit}
\ORR(U):= \frac{\PP (T=0 \given U )}{\PP (T = 1 \given U )}
\cdot
\frac{\PP(T=1)}{\PP(T=0)}
=e^{-\delta U}.
\#
Since $U$ is unbounded, the above odds ratio cannot be 
uniformly bounded by any constant. 
In this simple stylized example,  
this researcher cannot obtain any
informative range of the treatment
effects from the sensitivity analysis 
under a hypothesized marginal $\Gamma$-selection assumption~\eqref{eq:Tan2006}. 
\hfill 
$\square$
\end{example}

More generally, if a researcher concerns a 
scenario where the  
unmeasured confounding
is drastically severe in a small region of the sample space 
but non-exists in the remaining,  
it would require a very large, if not infinity,
value of $\Gamma$ for~\eqref{eq:Tan2006} 
to be practically meaningful. 
Sensitivity analysis under~\eqref{eq:Tan2006} 
thus provides a wide (thus uninformative) range 
of treatment effects.
However, since the magnitude of selection bias 
is large only in a small region, its overall impact
(on ATE for instance)
should still be small. 
A desirable sensitivity model should 
still produce informative bounds on the 
treatment effects in such situations, 
and more generally, capture the strength of unmeasured confounding 
beyond its maximum magnitude. 

We do mention that 
a few works in the literature~\citep{imbens2003sensitivity,franks2019flexible} postulate   
parametric models for treatment assignment 
that is affected by unobserved  confounders; 
such parametric models include the simple example discussed here as a special case and hence allow for unbounded local confounded effects. However, the key limitation is that the proposed confounding model is highly specialized to the logistic form, whereas the marginal $\Gamma$-selection criterion provides a non-parametric model that is quite general. 

Motivated by the merits of both worlds, 
in this paper, we develop 
a novel sensitivity model that 
describes the ``average''
strength of unmeasured confounding.  
We also develop a framework to conduct sensitivity analysis 
with our new models, which informs the range of treatment effects 
under various overall strength of unmeasured confounding.  
Our contributions are summarized in the following. 
\begin{itemize}
    \item \emph{A new sensitivity model}. We propose the $f$-sensitivity model, 
    a general, non-parametric model that characterizes the overall strength of 
    unmeasured confounding. 
    It is suitable for situations where the confounding 
    may be unbounded yet with a limited overall impact at the population level. 
    \item \emph{A new class of distributional robustness problems}.
    We show that the 
    partial identification bounds on 
    treatment effects under the $f$-sensitivity model
    can be represented by the solution to a class of
    DRO programs that are new to 
    the literature, providing a distributional robustness perspective 
    to sensitivity analysis under unmeasured confounding. 
    \item \emph{A new framework 
    for robust estimation and inference}. 
    We develop a set of 
    tools to estimate the optimal objective 
    of the new DRO problems; 
    the objective can be expressed via the solution to
    a weighted risk minimization problem, 
    with the unknown 
    weights determined by the covariate shift 
    between treatment and control groups.  
    We then propose estimators for the bounds  
    using a new debiasing technique 
    applied to the output of the corresponding 
    ERM problem. 
    We prove that our estimators are doubly-robust 
    to the estimation of nuisance components. 
    Furthermore, they
    enjoy an interesting one-sided validity 
    property that is specific to the partial identification 
    setting:  
    our estimators are still valid yet 
    perhaps conservative bounds when 
    the ERM step is completely off. 
\end{itemize}

%% file: model.tex

\subsection{The $(f,\rho)$-selection condition}

Our new $f$-sensitivity model is specified by
the following $(f,\rho)$-selection condition.

\begin{definition}[The $(f,\rho)$-selection condition]
\label{def:cond}
Suppose $f:\R_+ \mapsto \R$  is a convex function
such that $f(1) = 0$. 
Let $\ORR(x,u)$ be defined in~\eqref{eq:def_or}. 
$\PP$ satisfies the 
$(f, \rho)$-selection condition if for $\PP$-almost all $x$,
\#\label{eq:def_cond}
d_f(\PP) := \max\bigg\{
&\int f\big(\ORR(x,U) 
\big) 
\diff \PP_{U \mid X = x,T=1}, ~ 
\int f\big( \ORR(x,U)^{-1}\big) 
\big) \diff \PP_{U \given X = x,T=0}
\bigg\} \le \rho.
\#

\end{definition}

This new model addresses the unbounded confounding issue in Example~\ref{ex:motivation}: even though the odds ratio is not uniformly bounded, it is 
controlled overall; in this case, our new model 
can be a more reasonable description of the practical 
situation. 
We will discuss shortly about 
more settings where  
our method may be sensible. 
Now, let us first address the concern 
in Example~\ref{ex:motivation} using our framework. 

\addtocounter{example}{-1}
\begin{example}[Continued]\normalfont
We take $f(t) = t\log t$, a convex function with $f(1) = 0$. 
Continuing the computation in Example~\ref{ex:logit}, 
 the first term of $d_f(\PP)$ in Definition~\ref{def:cond} can be computed as 
\#\label{eq:ex_upper}
& \int f\big(\ORR(U) \big)\ud \PP_{U \given T = 1}
= \int -\delta U  \cdot e^{-\delta U}   
\ud \PP_{U \given T = 1}
= -\delta \int  u \cdot e^{-\delta u}   \cdot \frac{2e^{\delta u}}{1+e^{\delta u}} 
\cdot \frac{1}{\sqrt{2\pi}}e^{-u^2/2}\ud u <\infty.
\#
The right-handed side is approximately $0.2$ if we take $\delta = 1$ 
and $0.6$ if we take $\delta=2$. 
Note that $\int f\big(\ORR(U) \big)\ud \PP_{U \given T = 1}$
can be interpreted as the overall deviation of $\ORR(U)$ from $1$ 
in the treated (observed) group, which is bounded, even though $\ORR(U) \rightarrow \infty$ when $U \rightarrow -\infty$.
In this way, one could seamlessly use our framework
to conduct sensitivity analysis; this will inform the 
impact of the overall strength of unmeasured confounding 
on the treatment effects. \hfill $\square$ 
\end{example}

Two remarks on the $(f,\rho)$-selection condition
are in order.

\begin{remark}\normalfont
If $\PP$ satisfies the marginal $\Gamma$-selection 
condition~\eqref{eq:Tan2006}, then it automatically satisfies the $(f,\rho)$-selection condition with
any qualified $f$ and $\rho = \max\{f(1/\Gamma),f(\Gamma)\}$. 
This can easily be checked by noting that~\eqref{eq:Tan2006} implies
$ 
f ( \ORR(x,U)
 ) \le \max \{f(1/\Gamma),f(\Gamma) \} 
$ 
by convexity of $f$, 
and similarly by the symmetry of~\eqref{eq:Tan2006} 
$f ( \ORR(x,U)^{-1}
 ) \le \max \{f(1/\Gamma),f(\Gamma) \} $,
thereby leading to the bound on $d_f(\PP)$. 
It might appear that the $(f,\rho)$-selection condition 
is weaker than the marginal $\Gamma$-selection condition. 
However, we do note the two models give different characterizations, as for a distribution that 
satisfies marginal $\Gamma$-selection condition, 
it might satisfy $(f,\rho)$-selection condition 
for some $\rho$ that is much smaller than 
$\max\{f(1/\Gamma),f(\Gamma)\}$. 
\end{remark}
 
\begin{remark}\normalfont
In the definition of $d_f(\PP)$, we take the maximum 
of two integrals, each from one direction.
This is mainly to keep the condition symmetric 
with regards to the choice of the treated or control groups, 
and align with the convention in the sensitivity models in the literature~\citep{tan2006distributional}. 
That said, as we would see shortly in Section~\ref{subsec:distr_shift}, 
it might be more natural to only work with one of them (i.e. assume one of them be bounded by $\rho$)
when one of the 
counterfactuals is of primal interest. 
\end{remark}

\subsection{Comparison with other sensitivity models}
\label{subsec:interpret}

To better interpret the $(f,\rho)$-selection condition
and illustrate its difference from the (marginal) 
$\Gamma$-selection condition~\eqref{eq:Tan2006}, 
we provide a unified
perspective on the sensitivity models.
We first note a crucial property of $\ORR(X,U)$, 
a key quantity that characterizes the 
impact of unmeasured confounding. 
 
\begin{property}
\label{property:exp}
Let $\ORR(x,u)$ be defined in~\eqref{eq:def_or}. 
Then $\EE\big[\ORR(X,U) \given X,T=1 \big] = 1$ almost surely, where the conditional distribution 
is induced by the joint distribution of $(X,U,T)$. 
Also, $\ORR(X,U)=1$ holds 
$\PP_{X,U\given T=1}$-almost surely under the 
strong ignorability condition 
$T\,\indep\, (Y(1),Y(0))\given X$.
\end{property}
 
At a high level, both the $(f,\rho)$-selection
and the marginal $\Gamma$-selection condition
quantify how faraway the nonnegative mean-one 
random variable $\ORR(x,U)$ is from the constant one.
The marginal $\Gamma$-selection condition~\eqref{eq:Tan2006}
requires the maximum fluctuation of $\ORR(x, U)$ to be bounded within $[1/\Gamma, \Gamma]$ all the time.  
The $(f,\rho)$-selection condition, on the other hand,
characterizes the overall distance of $\ORR(x,U)$ from a constant. 
When we take $f(x) = \frac{1}{2}|x-1|$, the $(f,\rho)$-selection  
condition is similar to bounds on the total variation (TV) distance; 
when $f(x) = (x-1)^2$, the $(f,\rho)$-selection condition
resembles bounds on the $\chi^2$-distance
between $\ORR(x,U)$ and one. 
Different choices of $f$ pose different penalty 
for large values of confounding. 
For example, taking $f(x)=\frac{1}{2}|x-1|$, 
the contribution to the confounding measure 
is proportional to the absolute distance from $1$. 
For $f(x)=(x-1)^2$, 
the contribution to the confounding strength is larger for larger scale of confounding. 

We now illustrate the distinction between the 
$(f,\rho)$-selection condition and the 
marginal $\Gamma$-selection condition 
in two cases.
First,  
in the left panel of
Figure~\ref{fig:or} we see
three possible $\ORR(x,U)$ as functions of $U$,
all of which integrate to $1$ and with $U\sim \textrm{Unif}[0,1]$. 
There, the solid line is a constant function and indicates no 
unmeasured confounding. 
Among the other two, 
intuitively, 
the dotted curve has ``smaller'' confounding 
because most of the time
the odds ratio is quite close to  
$1$; one could imagine that 
in these regions, erroneously making 
the strong ignorability assumption may not 
incur too much bias. 
However, the upper bound on $\ORR(x,U)$ is 
large due to a small proportion of severe confounding 
at the left. 
In this case, although we imagine that 
the impact of $U$ for 
the dotted and dashed curves are drastically different, 
it requires the same $\Gamma$ in~\eqref{eq:Tan2006} 
to characterize them. 
As a result, partial identification bounds 
for treatment effects under the marginal $\Gamma$-selection 
condition may be uninformative;
a better measure 
for the confounding strength may instead be 
the \emph{overall fluctuation} of $\ORR(x,U)$ around $1$. 

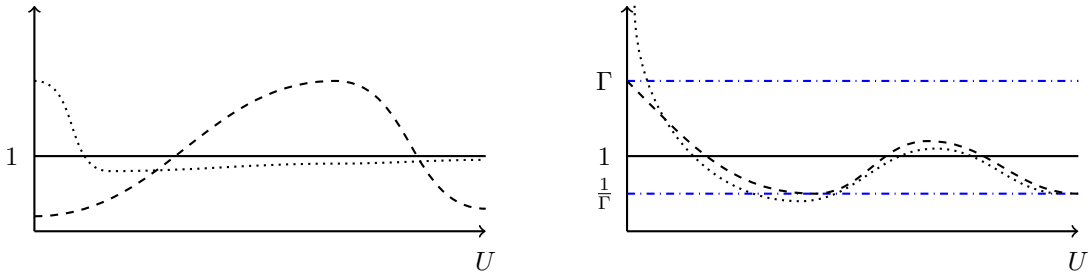
\begin{figure}[ht!]
\centering
\begin{minipage}{0.45\textwidth}
\centering
\begin{tikzpicture}
  \draw[->, thick] (-3, 0) -- (3, 0) node[right] {};
  \draw[->, thick] (-3, 0) -- (-3, 3) node[right] {};
  \draw[thick] (-3, 1) -- (3, 1) {};
  \node[] at (-3.3,1) {$1$};
  \node[] at (3,-0.4) {$U$};
    \coordinate (c1) at (-3,2);
    \coordinate (c2) at (-2,0.8);
    \coordinate (c3) at (1,0.9);
    \coordinate (c4) at (3,.95);
    \coordinate (b1) at (-3,.2);
    \coordinate (b2) at (1,2);
    \coordinate (b3) at (3,.3);
    \draw [dotted,thick] (c1) to[out=0,in=180] (c2) to[out=0,in=180](c3)
    to[out=0,in=180] (c4);
    \draw [dashed,thick] (b1) to[out=0,in=180] (b2) to[out=0,in=180](b3);
\end{tikzpicture} 
\end{minipage}
\hspace{0.2in}
\begin{minipage}{0.4\textwidth}
\centering
\begin{tikzpicture}
  \draw[->, thick] (-3, 0) -- (3, 0) node[right] {};
  \draw[->, thick] (-3, 0) -- (-3, 3) node[right] {};
  \draw[dash dot, blue, thick] (-3, 2) -- (3, 2)  {};
  \draw[dash dot, blue, thick] (-3, .5) -- (3, .5)  {};
  \node[] at (-3.3,1) {$1$};
  \node[] at (-3.3,2) {$\Gamma$};
  \node[] at (-3.3,.5) {$\frac{1}{\Gamma}$};
  \node[] at (3,-0.4) {$U$};
    \coordinate (c1) at (-3,2);
    \coordinate (c2) at (-.5,0.5);
    \coordinate (c3) at (1,1.2);
    \coordinate (c4) at (3,.5); 
    \coordinate (b1) at (-2.9,3);
    \coordinate (b2) at (-0.7,0.4);
    \coordinate (b3) at (1.1,1.1);
    \coordinate (b4) at (2.8,0.5);
    \draw [dashed,thick] (c1) to[out=315,in=180] (c2) to[out=0,in=180](c3)
    to[out=0,in=180] (c4);
    \draw [dotted,thick] (b1) to[out=270,in=180] (b2) to[out=0,in=180](b3)
    to[out =0, in =180](b4);
    \draw[thick] (-3, 1) -- (3, 1) {};
\end{tikzpicture} 
\end{minipage}
\caption{Left: examples of $\ORR(x,U)$ that are quite different but have similar upper bounds. Right: examples of $\ORR(x,U)$ that are similar but have drastically different upper bounds.}
\label{fig:or}
\end{figure}

The right panel of Figure~\ref{fig:or} 
plots another scenario where the uniform bound 
can be inaccurate. Here, the dotted and dashed lines 
describe two confounded cases where $\ORR(x,U)$ almost 
coincide except for the tail
at the left end. The dotted 
thus requires a much larger $\Gamma$
than the dashed one in the marginal sensitivity
model, if not infinity. 
In this case, because the treatment  probabilities 
(decided by the odds ratio) in these cases are 
so close and the tail region only takes 
a tiny part of the population, 
one could imagine the impact of confounding 
on the treatment effect to be close. 
Hence, besides the \emph{scale} 
of confounding, a sensitivity model 
should also takes into account the change
of having certain confounding strength.


Our $(f,\rho)$-selection condition exactly 
aims at resolving the above issues. 
For general choice of $f$, 
our sensitivity measure would give starkly different 
measures for the two cases in the left panel of Figure~\ref{fig:or}, while 
providing similar measures for those in the right panel. 
This is because it 
is an ``average'' measure of
the deviation of $\ORR(x,U)$ from the constant one
for strong ignorability. 
Correspondingly, 
sensitivity analysis from 
our model informs what would happen 
under a specific level of overall confounding strength.  

\subsection{Distributional shifts under the
$f$-sensitivity model}

\label{subsec:distr_shift}

The first observation in this paper relates the observables 
to the counterfactuals. 
We characterize the distributional shifts between the two 
under our $f$-sensitivity model, 
which identifies a new class 
of robust inference problems that, as far as we know, 
are new to the literature.  

We cast causal inference 
as a counterfactual inference problem:  
one needs to impute the missing outcome, i.e., 
the counterfactual, to estimate treatment effects 
at the population level. 
For example, to estimate the ATC:
$\EE[Y(1)\given T=0]-\EE[Y(0)\given T=0]$, 
one needs to impute the first term, the counterfactual 
mean of $Y(1)$ in the control group.  
The distribution
of the unobservable $(X,Y(1))$ in the control group is 
\$
\PP_{X,Y(1)\given T=0} = \underbrace{\PP_{X\given T=0}}_{\rm (a)} 
\times \underbrace{\PP_{Y(1)\given X,T=0}}_{\rm (b)}.
\$
Here part (a) is identifiable from the observations, 
but part (b) is not when there is unmeasured confounding.   
Our first result states that under the $(f,\rho)$-selection condition, 
(b) is bounded from 
its counterpart in the observable
in terms of $f$-divergence.

\begin{definition}[$f$-divergence]
Let $\PP$ and $\QQ$ be two probability distributions over a space $\Omega$ such that $\PP$ is absolutely continuous with respect to $\QQ$. For a convex function $f$ such that $f(1)=0$, the $f$-divergence of $\PP$ from $\QQ$ is defined as $D_f(\PP\,\|\, \QQ) = \EE_{\QQ}\big[f(\frac{\ud \PP}{\ud \QQ})\big]$, where $\frac{\ud \PP}{\ud \QQ}$ is the Radon-Nikodym derivative. 
\end{definition}

Popular examples for $f$-divergence in the literature include 
the Kullback–Leibler (KL) divergence with $f(t) = t\log t$, 
the total variation (TV) distance with $f(t)=|t-1|/2$, 
Pearson $\chi^2$-divergence with $f(t) = (t-1)^2$, and the 
Cressie-Read family of $f$-divergences~\citep{cressie1984multinomial} 
parametrized by $k$, where 
$f_k(t):= \frac{t^k - kt + k-1}{k(k-1)}$. 
Throughout this paper, we work with generic forms of $f$-divergence, 
and provide discussions on concrete examples where proper conditions are satisfied 
for our analysis. 

\begin{lemma}\label{lem:bound_shift_single}  
Under the $(f,\rho)$-selection condition, we have 
\$
D_f\big(  \PP_{Y(1)\given X=x,T=0} \,\|\, \PP_{Y( 1)\given X=x,T= 1}   \big) \leq \rho 
\$
for $\PP_{X\given T=1}$-almost all $x$; 
that is, the $f$-divergence between the conditional distributions 
in the two groups are bounded by $\rho$ for almost all $X$ in group $T=1$. 
\end{lemma}
\begin{proof}[Proof of Lemma~\ref{lem:bound_shift_single}]
Suppose a distribution $\PP$ over $(X,U,T,Y(0),Y(1))$
satisfies the $(f,\rho)$-selection condition. 
By condition~\eqref{eq:u_confoundedness} and the data-processing inequality,
\$
D_f\big(\PP_{Y(1) \given X = x, T=0} 
~\|~ \PP_{Y(1) \given X = x, T=1}\big)
\le & 
D_f\big(\PP_{Y(1),U \given X = x, T=0} 
~\|~ \PP_{Y(1),U \given X = x, T=1}\big)\\
= & 
\E_{Y(1),U \given X = x,T=1}
\bigg[f\Big(\frac{\ud \PP_{Y(1),U \given X = x,T = 0}}
{\ud \PP_{Y(1), U \given X=x, T=1}}\Big)\bigg].
\$
We note that the likelihood ratio can be decomposed as
\$
\frac{\ud \PP_{Y(1),U \given X = x,T = 0}}
{\ud \PP_{Y(1), U \given X = x, T=1}}
&= \frac{\ud \PP_{Y(1) \given U, X = x,T = 0}}
{\ud \PP_{Y(1) \given U, X = x, T=1}}
\cdot
\frac{\ud \PP_{U \given X = x,T = 0}}
{\ud \PP_{U \given X = x, T=1}} \\
&\stackrel{\rm (a)}{=}  
 \frac{\ud \PP_{U \given X = x,T = 0}}
{\ud \PP_{U \given X = x, T=1}} 
=  \frac{\PP(T = 0 \given X = x, U)}
{\PP(T = 1 \given X = x, U )}
\cdot
\frac{\PP(T = 1 \given X = x)}
{\PP(T = 0 \given X = x)},
\$
where step (a) is due to
condition~\eqref{eq:u_confoundedness}.
Combining the above two facts yields 
\$
D_f\big(\PP_{Y(1) \given X, T=0} 
~\|~ \PP_{Y(1) \given X, T=1}\big)
\le & \E_{Y(1),U \given X , T=1}
\bigg[f\Big(\frac{\PP(T = 0 \given X, U)}
{\PP(T = 1 \given X, U )}
\cdot
\frac{\PP(T = 1 \given X)}
{\PP(T = 0 \given X)}\Big)\bigg]\\
= & \E_{U \given X , T=1}
\bigg[f\Big(\frac{\PP(T = 0 \given X, U)}
{\PP(T = 1 \given X, U )}
\cdot
\frac{\PP(T = 1 \given X)}
{\PP(T = 0 \given X)}\Big)\bigg]
\le \rho
\$
almost surely, 
where the last inequality is due to the 
$(f,\rho)$-selection condition.
\end{proof}
 
We complete our characterization of 
the counterfactual distributions $\PP_{X,Y(1)\given T=0}$ 
induced by all super-populations that 
agrees with the observables 
and satisfies our sensitivity models.

\begin{proposition}\label{prop:robust_set}
Let $\PP^{\sup}$ be the true
unknown super-population over $(X,U,T,Y(0),Y(1))$
and let $\cP$ be the set of all 
distributions over $(X,U,T,Y(0),Y(1))$. 
Let $\PP^{\obs}_{X,Y,T}$ be the joint distribution 
of all observable random variables $(X,Y,T)$.  
Define $\cQ_{1,0}$ to be the 
ambiguity set of all 
counterfactual distributions that 
agrees with the observables and satisfies the $(f,\rho)$ selection condition, 
i.e.,
\$
\cQ_{1,0} = \big\{ \PP^{ }_{X,Y(1)\given T=0}  
\colon \PP\in \cP,~ \PP^{ }_{X,Y,T} = \PP^{\obs}_{X,Y,T},
~ \PP^{ }\,\text{satisfies Definition~\ref{def:cond}} \,  \big\}.
\$
Then  
$\PP^{\sup}_{X,Y(1)\given T=0} \in \cQ_{1,0}$, and 
\#\label{eq:large_idset}
\cQ_{1,0} \subset  \Big\{ \QQ \colon  {\textstyle \frac{\ud \QQ_X}{\ud \PP_{X\given T=1}^{\obs}}}(x) = r_{1,0}(x),
  ~&D_f\big( \QQ_{Y\given X=x} ~\big\|~  \PP_{Y\given X=x,T=1}^{\obs} \big) \leq \rho,
~\text{for }\PP^\obs_{X\given T=t}\text{-almost all }x  \Big\} ,
\#
where $r_{1,0}(x) = \frac{(1-e(x))p_1}{e(x)(1-p_1)}$, 
and 
$e(x) = \PP^\obs(T=1\given X=x)$, $p_1=\PP^\obs(T=1)$.  
\end{proposition}

We defer the proof of Proposition~\ref{prop:robust_set}
to Appendix~\ref{app:proof_bound_shift}, where we 
prove a stronger version 
that gives a tight characterization of $\cQ_{1,0}$.
Symmetrically, we can also define $\cQ_{0,1}$ as 
the identification set of $\PP_{X,Y(0) \given T=1}$; the tight 
characterization of $\cQ_{0,1}$ is also given 
in Appendix~\ref{app:proof_bound_shift}. 
From now on, we  
only consider the ambiguity sets in Proposition~\ref{prop:robust_set} 
to emphasize the 
more general 
distributional robustness aspect of this problem. 

Proposition~\ref{prop:robust_set} identifies a new class 
of robust inference problems, where the target distribution 
has an identifiable $X$-shift and the unidentifiable
conditional distribution is restricted in $f$-divergence
ball; 
it is similar to the ambiguity set studied in~\citet{jin2021sensitivity}. 
This model is closely related to, but quite distinct from 
other ambiguity sets involving $f$-divergence  
in the literature~\citep{duchi2021learning,si2020distributional,andrews2020informativeness} 
that only concern marginal (joint) distributions.

\begin{remark}[Relation to other $f$-divergence bounds] \normalfont
Previous works in the literature (see Section~\ref{subsec:related_work} for a summary) often work  under the $f$-divergence ball around 
the marginal distribution $\PP_{X,Y}$, 
characterized by 
\#\label{eq:id_set_lit}
\tilde\cQ = \big\{ \QQ \colon    D_f\big( \QQ_{X,Y} \big\|  \PP_{X,Y}  \big)  \ \leq \rho  \big\}.
\#
While in our formulation, the ambiguity sets take the form 
\#\label{eq:id_set_ours}
\begin{array}{l}
\cQ = \Big\{ \QQ \colon  \frac{\ud \QQ_{X}}{\ud \PP_X}(x) =  r(x),~  D_f\big( \QQ_{Y\given X} \big\|  \PP_{Y\given X}  \big)  \ \leq \rho  \Big\}, 
\end{array}
\#
where $r(x)$ is a known or identifiable function. 
Such distinction is similar to 
what has been observed in~\citet[Remark 3]{jin2021sensitivity}.
Instead of bounding the overall shift
  as in~\eqref{eq:id_set_lit}, the constraint in~\eqref{eq:id_set_ours}
  actually allows freedom in the shift of $X$: the
  sets in~\eqref{eq:id_set_ours} can be small as long as $\rho$ is small. 
  For counterfactual inference under the strong ignorability
  condition, 
  the set~\eqref{eq:id_set_ours} can be a singleton even if $\PP_X$ and
  $\TPP_X$ are drastically different, while \eqref{eq:id_set_lit}
  might require a large $\rho$ to hold.  
  More generally, when there is
  a known (or estimable) large shift in  
  $\PP_X$ but a relatively small shift in  
  $\PP_{Y\given X}$, \eqref{eq:id_set_ours} provides a tighter range of the
  target distributions, and the methods we develop in this paper 
  can be directly applied. 
\end{remark}

\subsection{Related work}
\label{subsec:related_work}

This work falls within a strand of 
sensitivity analysis 
that models the impact of unmeasured confounding 
through bounds on selection bias. 
We summarize a few of them that are not mentioned 
above. 
Remarkably, 
\citet{rosenbaum1983assessing}  
studies the impact of selection bias among matched pairs, 
which is further extended by a series of 
works~\citep{rosenbaum1987sensitivity,gastwirth1998dual,rosenbaum2002attributing,rosenbaum2002observational}
to sensitivity models that uniformly 
bounds the selection bias 
among 
samples with matching
covariates.  
Also related to sensitivity analysis under uniformly bounded 
selection bias is~\citet{yadlowsky2018bounds}, which works under 
an extention of Rosenbaum's sensitivity model 
that is similar to~\citet{tan2006distributional}. 
Besides modelling the selection bias among treatment, observed covariates and 
unmeasured confounders, 
we mention in passing that~\citet{ding2016sensitivity} considers the sensitivity analysis 
of other metrics than treatment effects. 

We develop our framework based on observing
the distributional shift 
between the observations and the counterfactuals. 
This perspective 
echos the ideas of several previous works~\citep{jin2021sensitivity,yadlowsky2018bounds,dorn2021doubly} on sensitivity analysis. 
In particular, we  formulate the estimand 
via an optimization problem under constraints on 
distributional shifts, similar 
to~\citet{yadlowsky2018bounds,dorn2021doubly}. 
However, as we work under different sensitivity models, 
both the form of distributional shifts 
and the techniques for statistical 
estimation and inference are distinct.

$f$-divergence is often used to 
characterize discrepancy between distributions~\citep{renyi1961measures,morimoto1963markov,csiszar1964informationstheoretische,liese2006divergences,rahman2016f}.  
In our work, we use a quantity similar to 
$f$-divergence to measure 
the deviation of the odds ratio  
from $1$, hence characterizing the overall magnitude of 
the 
selection bias caused by unmeasured confounding. 
This in turn leads to 
the bounded $f$-divergence between 
the conditional distribution of the 
counterfactuals and that of the observations 
under our model, 
while the covariate distributional shift is 
identifiable from the data.  
To the best of our knowledge, this type of 
distributional shifts have not been studied before.  

By connecting sensitivity analysis 
to distributionally robust optimization, 
this works is also related to 
a line of work on 
estimation, inference and learning 
under various types of distributional shifts, e.g., $f$-divergence, 
outside 
the task of inferring causal effects. 
Among them,~\citet{christensen2019counterfactual} 
places bounds on the marginal distributions of hidden variables 
in  structural equation models;~\citet{andrews2020informativeness} studies parameter estimation 
when the joint distributions of variables are 
within an $f$-divergence ball;~\citet{duchi2021learning} 
studies  
the empirical risk minimization problem 
when the joint distribution 
of $(X,Y)$ shifts within an $f$-divergence ball; 
\citet{si2020distributional} studies the policy learning 
for contextual bandits under unknown marginal distribution shifts; 
\citet{gupta2021s} studies 
the estimation and inference 
of statistical parameters under distributional shifts 
in certain directions, etc.  
The new class of robust inference problems in our work
is different from the settings of these works.  
We also develop a set of tools for estimation and inference 
under this new type of distributional shifts.

Finally, because 
the covariate shifts need to be estimated, 
we propose a debiasing technique 
to obtain root-$n$ rate of inference 
when the nuisance components are estimated with 
slower convergence rates. 
This is also generally connected 
to a vast body of missing data literature with 
unknown missing mechanisms, although 
in different contexts and with different details. 
In particular, 
we use the cross-fitting~\citep{schick1986asymptotically,zheng2011cross,chernozhukov2018double} technique 
to mitigate the  error in nuisance component estimation; 
bias correction using a different dataset 
under covariate shift  
is also related to and inspired by 
the transductive inference technique in~\cite{jin2021one}.

%% file: sec3.tex

\subsection{Bounds on the treatment effects}
\label{subsec:opt}

In this part, we study the (population-level) partial identification bounds 
on counterfactual means under our sensitivity model. 
They are provided as 
solutions to convex optimization problems 
that only involve identifiable quantities. 
We consider  $\EE[Y(1)\given T=0]$ for illustration.

\begin{proposition}\label{prop:convex_prob}
Let $\mu_{1,0}^-$ (resp.~$\mu_{1,0}^+)$ 
be the optimal objective function 
of the convex optimization problem 
\#\label{eq:opt_cond}
\mathop{\min (\text{resp.~}\max)}_{L(x,y) \textnormal{~measurable}}~
             &\E\big[Y(1) L(X,Y(1)) \biggiven T=1\big]\\
  \textnormal{s.t.}~&\E[L(x,Y(1)) \given X=x, T=1]= r_{1,0}(x) \\
                    &\E\big[f\big(L(x,Y(1)) /r_{1,0}(x)  \big)\biggiven X=x,T=1\big]\le \rho,\quad \mbox{for almost all }x,
\#
where all the expectations are induced by the observed distribution. 
Then  
$\mu_{1,0}^- \leq \EE[Y(1)\given T=0] \leq \mu_{1,0}^+$ 
under the $(f,\rho)$-selection condition.  
\end{proposition}

As we have discussed, the counterfactual means 
are the building blocks for treatment effects. 
Proposition~\ref{prop:convex_prob} immediately 
implies bounds on the ATC:
denote the observable group-wise means as $\mu^\obs_t := \EE[Y(t)\given T=t]$ 
for $t\in\{0,1\}$, then 
under the $(f,\rho)$-selection condition, the ATC is bounded as 
\# 
\mu_{1,0}^- - \mu^\obs_0 \leq \EE[Y(1)-Y(0)\given T=0] \leq \mu_{1,0}^+ - \mu^\obs_0.\label{eq:bound_atc}
\#
Switching the role of $1$ and $0$ in Proposition~\ref{prop:convex_prob},
one can obtain bounds on $\EE[Y(0) \given T=1]$; let $\mu^+_{0,1}$
and $\mu^-_{0,1}$ denote the upper and lower bound on 
$\EE[Y(0) \given T=1]$, respectively,
we then 
get bounds on the ATT: 
\$
\mu^\obs_1 - \mu_{0,1}^+ \leq \EE[Y(1)-Y(0)\given T=1] \leq \mu^\obs_1 - \mu_{0,1}^-.
\$
By the decomposition of ATE (average treatment effects), we 
also have the representation of lower and upper bounds for 
$\EE[Y(1)-Y(0)]$. 
Under the $(f,\rho)$-selection condition, we have 
\#\label{eq:bound_ate}
p_1\big(\mu^\obs_1 - \mu_{0,1}^+\big) 
+ p_0\big( \mu_{1,0}^- - \mu^\obs_0 \big)
\leq \EE[Y(1)-Y(0)] \leq p_1\big(\mu^\obs_1 - \mu_{0,1}^-\big) 
+ p_0\big( \mu_{1,0}^+ - \mu^\obs_0 \big).
\#
Estimation of these bounds thus boil down to 
that of $\mu_{t,1-t}^{\pm}$ 
under our sensitivity models, 
the 
optimal objective value of the convex optimization problems
in Proposition~\ref{prop:convex_prob}.

\begin{remark} \normalfont
As we mentioned before, optimal objectives of 
the problems in Proposition~\ref{prop:convex_prob} 
are not necessarily tight bounds for counterfactual means. 
To align with the literature and keep 
a relatively clean formulation of dual problems, 
we only account for the direction 
$D_f(\PP_{Y(1)\given X=x,T=0}~\|~\PP_{Y(1)\given X,T=1})\leq \rho$ 
when considering $\EE[Y(1)\given T=0]$.  
For completeness, we discuss 
the tight bounds on 
counterfactual means, hence ATT and ATC, 
in Section~\ref{sec:discussion}. 
We also note that 
combining  sharp bounds on ATT and ATC 
does not necessarily lead to sharp bounds on 
ATE, as they might be attained by different super-populations. 
We leave the investigation of sharp bounds on ATEs for future pursuit. 
\end{remark}

\subsection{From the primal to the dual} 
It is hard to directly solve 
the infinite-dimensional optimization problem~\eqref{eq:opt_cond}. 
We address this issue
by translating to its dual form which might be easier to tackle. 
In the following, we 
primarily focus on $\mu_{1,0}^-$, the lower bound on 
$\EE[Y(1)\given T=0]$, and the 
same idea carries over to the upper bounds 
as well as bounds for other quantities.  
Proposition~\ref{prop:dual_cond} 
represents $\mu_{1,0}^-$ via a 
dual formulation, 
whose proof 
is in Appendix~\ref{app:subsec_dual_cond}.
\begin{proposition}\label{prop:dual_cond}
The optimal objective of~\eqref{eq:opt_cond} is given by 
\#\label{eq:dual_cond}
\mu_{1,0}^- 
&= - \inf_{\alpha(X)\geq 0, \eta(X)\in \RR} \EE\bigg[ r_{1,0}(X) \Big\{  \alpha (X)  f^*\Big( \frac{Y(1) + \eta (X) }{- \alpha (X)} \Big) + \eta (X)  + \alpha (X)   \rho \Big\} \bigggiven  T=1\bigg] ,
\#
where $f^*(s)=\sup_{t\geq 0} \{st - f(t)\}$ is the conjugate function of $f$. In particular, 
denoting $\ell(\alpha,\eta,x,y) = \alpha f^*(\frac{y+\eta}{-\alpha}) + \eta + \alpha \rho$ for 
$(\alpha,\eta)\in \RR^+\times \RR$, we have 
$\mu_{1,0}^- =  - \EE\big[ \ell(\alpha^*(X),\eta^*(X),X,Y(1)) \biggiven  T=1\big],
$
where  for $\PP_{X\given T=1}$-almost all $x$, 
\#\label{eq:perx_opt}
\big( \alpha^*(x),\eta^*(x) \big) \in \argmin{\alpha\geq 0, \eta\in \RR} ~\EE\bigg[  \alpha  f^*\Big( \frac{Y(1) + \eta }{- \alpha} \Big) + \eta + \alpha  \rho \bigggiven X=x,T=1\bigg].
\# 
\end{proposition}

Our task is then 
to estimate the dual formulation~\eqref{eq:dual_cond}, which can be 
viewed as a risk minimization 
problem. 
However, in sharp contrast to 
ERM problems in the literature, 
it involves a typically unknown 
weight $r_{1,0}(x)$ which depends on the propensity score 
$e(x)=\PP(T=1\given X=x)$. 
As the estimation rate of this quantity is often slower than root-$n$, 
directly solving~\eqref{eq:dual_cond} with plug-in weights 
might yield 
inaccurate estimators and prohibit root-$n$ statistical inference. 

To address this issue, we will 
make use of the second observation 
in Proposition~\ref{prop:dual_cond}:  
$(\alpha^*(x),\eta^*(x))$, the optimizer for~\eqref{eq:dual_cond}, 
is also 
the minimizer of per-$x$ conditional risk.  
This property crucially allows us to 
estimate $\alpha^*(\cdot)$ and $\eta^*(\cdot)$ 
without knowledge of $r_{1,0}(\cdot)$. 
Our general idea is to employ  empirical 
risk minimization tools to estimate $\alpha^*(\cdot),\eta^*(\cdot)$, 
and then estimate $\mu_{1,0}^-$ by plugging into~\eqref{eq:dual_cond}. 
Still, the slow estimation rate for 
the weights and the optimizers 
poses additional challenges for statistical inference. 
We will also develop a novel adjustment technique to still 
achieve root-$n$ inference 
when these quantities 
are estimated at a slow rate.

Before introducing our procedures, 
we present a result on the behavior of the optimizer $\alpha^*(x)$: 
it is positive 
as long as $\PP_{Y(1)\given X,T=1}$ does not have a large point mass 
at its essential infimum and the function $f$ satisfy some 
regularity conditions in the limit. 
The proof of Proposition~\ref{prop:positive_alpha} is deferred to Appendix~\ref{app:proof_positive_alpha}.

\begin{proposition}\label{prop:positive_alpha} 
Define $\uly(x) = \sup\big\{t: \PP(Y(1) < t \given X = x,
T=1) = 0 \big\}$ and
$\bar p(x) = \PP\big(Y(1) = \uly(x) \given X = x, T = 1\big)$.
We assume that $\bar{p}(x)f(1/\bar{p}(x)) + (1-\bar p(x)) f(0) > \rho$ 
for $\PP_{X \given T=1}$-almost all $x$. 
Also suppose there exist constants $L$ and $U$
such that  $f(x)^* \ge L$ for $x \in \RR$,
$f^*(x) \le U$ for $x \le 0$,  
$\lim_{ x\rightarrow -\infty}f^*(x)/x = 0$ and $\lim_{ x\rightarrow \infty}f^*(x)/x = \infty$.
Then the 
solution to~\eqref{eq:dual_cond}
satistifies $\alpha^*(x) > 0$ for $\PP_{X\given T=1}$-almost
all $x$.
\end{proposition}

In particular,  
the conditions on $f$ hold for a large variety of functions; 
concrete examples include KL divergence, where $f(x) = x\log x$ and $f^*(x) = e^{x-1}$, 
as well as $\chi^2$-divergence, where $f(x)=(x-1)^2$ and $f^*(x) = \frac{1}{4}((x+2)_+^2-1)$, etc. 

In the following, we assume throughout that 
the conditions of Proposition~\ref{prop:positive_alpha} hold, 
hence by the compactness of $\cX$, 
there exists some $\epsilon>0$ such that 
$\alpha^*(x)>\epsilon$ for $\PP_{X\given T=1}$-almost
all $x$ . 
An important implication is that, 
as $\alpha^*(x)$ lies in the interior of $[0,\infty)$, 
the gradient of the risk function is typically mean-zero 
at $(\alpha^*(x),\eta^*(x))$; 
this would play an important role in 
achieving root-$n$ statistical inference. 

\subsection{The estimation procedure}
\label{subsec:cond_method}

We start with splitting 
samples in the treated and control groups 
into three equally sized folds, denoted as 
$\cI_1^{(j)},\cI_0^{(j)}$, $j=1,2,3$, respectively. 
For each $j=1,2,3$, 
we use samples in $\cI^{(j+1)}_1$ and $\cI_0^{(j+1)}$ 
to obtain an estimator 
$\hat{r}^{(j)}$ for $r_{1,0}$, 
and  
solve an empirical risk minimization (ERM) problem 
to obtain estimators $\hat\alpha^{(j)}$ and $\hat\eta^{(j)}$ 
for $(\alpha^*,\eta^*)$ without knowledge of $r_{1,0}$; 
this empirical risk minimization step will be discussed shortly after.  
We then define the function 
\$
\hat{H}^{(j)}(x,y) = \hat\alpha^{(j)}(x)  f^*\Big( \frac{y + \hat\eta^{(j)}(x) }{- \hat\alpha^{(j)}(x)} \Big) + \hat\eta^{(j)}(x) + \hat\alpha^{(j)}(x) \rho.
\$ 
Using data in $\cI_1^{(j+2)}$, we run a regression algorithm 
to obtain an estimator $\hat{h}^{(j)}$ for 
$
\bar{h}^{(j)}(x) := \EE\big[\hat{H}^{(j)}(X,Y(1)) \biggiven  X=x,T=1,\cI_1^{(j+1)} \big], 
$
where we view $\hat\alpha^{(j)}$ and $\hat\eta^{(j)}$, hence $\hat{H}^{(j)}$, 
as fixed functions (e.g., by conditioning on $\cI_1^{(j+1)}$ and $\cI_0^{(j)}$. 
Finally, we define the estimator 
\$
\hat\mu_{1,0}^{(j)} = \frac{1}{|\cI_1^{(j)}|} \sum_{i\in \cI_1^{(j)}}
\hat{r}^{(j)}(X_i) \big(\hat{H}^{(j)}(X_i,Y_i) - \hat{h}^{(j)}(X_i)\big) + \frac{1}{|\cI_0^{(j)}|} \sum_{i\in \cI_0^{(j)}} \hat{h}^{(j)}(X_i).
\$
The above procedure is repeated for each $j=1,2,3$, 
and we average the three estimators to obtain
\$
\hat\mu_{1,0}^- =  -\frac{1}{3} \sum_{j=1}^3 \hat\mu_{1,0}^{(j)}.
\$
The whole procedure is summarized in Algorithm~\ref{alg:est}. 
In the algorithm, we refer to $\cI^{(k)}_1$ for some $k >3$ 
as $\cI^{(k \text{ mod } 3)}_1$; the same principle applies to $\cI_0^{(k)}$. 
We let $\cR(\cdot,\cdot)$ to be a generic algorithm 
that uses data in $\cI_1$ and $\cI_0$ to obtain an estimator 
$\cR(\cI_1,\cI_0)$ for $r$. 
We use $\textsf{ERM}(\cdot)$ to denote a generic ERM algorithm 
that uses any data  $\cI $ to output an estimator 
$\textsf{ERM} (\cI )$ for $(\alpha^*(\cdot),\eta^*(\cdot))$. 
We let \textsf{Reg}$(\cdot)$ denote a generic regression algorithm 
that takes data $(X_i,Z_i)_{i\in \cI}$ to 
output an estimator \textsf{Reg}$(Z_i\sim X_i,~i\in \cI)$
for $\EE[Z\given X=\cdot]$.

\begin{algorithm}[htbp]
  \caption{Estimation procedure for $\mu^-_{1,0}$}\label{alg:est}
\begin{algorithmic}[1]
\REQUIRE Treated samples $\cI_1$; control samples $\cI_0$; 
the algorithm $\cR$ for
estimating $r$; the ERM algorithm $\textsf{ERM}(\cdot)$ for estimating 
$(\alpha, \eta)$; the regression algorithm \textsf{Reg}$(\cdot)$ for 
obtaining $\bar h$.
\vspace{0.05in}
\STATE Randomly split $\cI_1$ and $\cI_0$ into three equal-sized
groups: $\cI_1^{(j)}, \cI_0^{(j)}$, $j=1,2,3$. 
\FOR{$j = 1,2,3$}
\STATE Estimate $r_{1,0}$: obtain $\hat{r}^{(j)}(\cdot) \leftarrow \cR\big(\cI_1^{(j+1)}, \cI_0^{(j+1)}\big)$.
\STATE Estimate $(\alpha^*,\eta^*)$: obtain  $(\hat{\alpha}^{(j)}(\cdot),\hat{\eta}^{(j)}(\cdot)) \leftarrow \textsf{ERM}\big(\cI_1^{(j+1)}\big)$
\STATE Conditional regression: $\hat{h}^{(j)}(\cdot) \leftarrow \textsf{Reg}(\hat{H}^{(j)}(X_i,Y_i) \sim X_i,~i \in \cI_1^{(j+2)})$.
\STATE Compute $
\hat\mu_{1,0}^{(j)} \leftarrow \frac{1}{|\cI_1^{(j)}|} \sum_{i\in \cI_1^{(j)}}
\hat{r}^{(j)}(X_i) \big(\hat{H}^{(j)}(X_i,Y_i) - \hat{h}^{(j)}(X_i)\big) + \frac{1}{|\cI_0^{(j)}|} \sum_{i\in \cI_0^{(j)}} \hat{h}^{(j)}(X_i).
$ 
\ENDFOR
\vspace{0.05in}
\ENSURE Estimator $\hat{\mu}^-_{1,0} =  -\frac{1}{3} \sum_{j=1}^3 \hat\mu_{1,0}^{(j)}$.
\end{algorithmic}
\end{algorithm}

In Algorithm~\ref{alg:est}, 
the subroutines $\cR(\cdot,\cdot)$ and \textsf{Reg}$(\cdot)$ 
are standard: to estimate $r_{1,0}$, one could use 
$\cI_1^{(j+1)}\cup\cI_0^{(j+1)}$ to estimate the propensity score 
$e(x)$ with any regression algorithm, and then plug in the definition of $r_{1,0}$. 
Similarly, \textsf{Reg}$(\cdot)$ can be any regression 
algorithm that fits a conditional mean function 
given i.i.d.~data. 
Widely adopted regression methods in the literature 
include localized nonparametric methods like kernel regression \citep{nadaraya1964estimating,watson1964smooth}, local polynomial regression \citep{cleveland1979robust,cleveland1988locally}, smoothing spline \citep{green1993nonparametric} and modern machine learning methods including regression trees \citep{Breiman:decisionTree} and random forests \citep{ho1995random}, to name a few.
The ERM step is relatively unique to our problem, 
and we discuss it in more details with rigorous guarantees as follows. 

\subsection{Solving for $\hat\alpha(\cdot)$ and $\hat\eta(\cdot)$ } 
\label{subsec:cond_sieve}

We take a moment to elaborate on 
the estimation of $\hat\theta^{(j)}:=( \hat\alpha^{(j)},\hat\eta^{(j)})$. 
From Proposition~\ref{prop:dual_cond}, $\theta^* := (\alpha^*, \eta^*)$ 
is also 
the population risk minimizer of $\EE[\ell(\theta,X,Y(1))\given T=1 ]$ (i.e., removing $r_{1,0}(x)$), 
where the loss function 
\$
\ell(\theta,x,y) = \alpha(x)  f^*\Big( \frac{y + \eta (x) }{- \alpha (x)} \Big) + \eta (x)  + \alpha (x)   \rho
\$
is convex in $\theta = (\alpha,\eta)$. 
The empirical risk is correspondingly (recall that we run 
ERM with fold $\cI_1^{(j+1)}$)
\$
\hat\EE_n \big[\ell(\theta,X,Y(1))\big] = \frac{1}{|\cI_1^{(j+1)}|} \sum_{i\in \cI_1^{(j+1)}} \ell(\theta, X_i,Y_i).
\$ 
We can thus consider a function class $\Theta$, and solve for 
the empirical risk minimization (ERM) problem. 
This approach is similar to~\citet{yadlowsky2018bounds}; 
however, 
they express the bounds of conditional expectations 
of counterfactuals themselves as empirical risk minimizers, 
while we use this ERM step as an intermediate step and 
employ distinct downstream techniques. 
To solve this ERM problem, 
we use the method of sieves~\citep{geman1982nonparametric};
we consider an increasing sequence $\Theta_1\subset \Theta_2 \subset \cdots$ 
of spaces of smooth functions, and let 
\$
\hat\theta^{(j)}  = \argmin{\theta\in \Theta_n}~ \hat\EE_n \big[\ell(\theta,X,Y(1))\big].
\$
We consider two examples of seives inspired by~\citet{yadlowsky2018bounds}. 

\begin{example}[Polynomials] \normalfont
    \label{ex:poly}
Let 
Pol$(J)$ be the space of $J$-th order polynomials on $[0,1]$:
\$
\textrm{Pol}(J,\epsilon) = \Big\{ x\mapsto  {\textstyle  \sum_{k=0}^J} a_k x^k   \colon a_k \in \RR  \Big\},
\$
and let 
Pol$(J,\epsilon)$ be the space of $J$-th order polynomials on $[0,1]$ 
truncated at $\epsilon>0$: 
\$
\textrm{Pol}(J,\epsilon) = \Big\{ x\mapsto \max\{ \epsilon, {\textstyle  \sum_{k=0}^J} a_k x^k \big\} \colon a_k \in \RR  \Big\}.
\$
Then we define the sieve $\Theta_n = \Theta_n^\alpha \times \Theta_n^\eta$, 
where $\Theta_n^\alpha = \{ x\mapsto \prod_{k=1}^d f_k(x_k)\colon f_k \in \mathrm{Pol}(J_n,0),k=1,\dots,d\}$ and $\Theta_n^\eta = \{ x\mapsto \prod_{k=1}^d f_k(x_k)\colon f_k \in \mathrm{Pol}(J_n),k=1,\dots,d\}$ for $J_n\to \infty$. 
\end{example}

Compared to~\citet{yadlowsky2018bounds}, our function class 
additionally truncates the functions away from zero for $\alpha(x)$: 
we note that, if 
$\alpha^*(x)$ is always positive (implied by the minimality of the risk function 
and Proposition~\ref{prop:positive_alpha}) and continuous 
(satisfied if $\PP_{Y(1)\given X=x,T=1}$ is smooth in $x$)
and $\cX$ is a compact set, then 
there exists a positive $\epsilon>0$ such that  $\inf_{x\in \cX}\alpha^*(x)\geq \epsilon$. In practice, we can set $\epsilon$ to be small enough, 
or let $\epsilon = \epsilon_n$ decays slowly to zero; 
this does not hurt the capability of function class 
or the convergence rates when $n$ is sufficiently large. 

\begin{example}[Splines] \normalfont
    \label{ex:spl}
Let $0=t_0<\dots<t_{J+1}=1$ be knots that satisfy  
$\frac{\max_{0\leq j\leq J}(t_{j+1}-t_j)}{\min_{0\leq j\leq J (t_{j+1}-t_j)}} \leq c$ for some $c>0$. We define 
the space for $r$-th order splines with $J$ knots as 
\$
\textrm{Spl}(r,J) = \Big\{  x\mapsto {\textstyle \sum_{k=0}^{r-1} a_k x^k + \sum_{j=1}^J b_j (x-t_j)_+^{r-1}}\colon a_k, b_k \in \RR   \Big\}
\$
and the truncated space for $r$-th order splines with $J$ knots as 
\$
\textrm{Spl}(r,J) = \Big\{  x\mapsto \max\big\{ \epsilon, {\textstyle \sum_{k=0}^{r-1} a_k x^k + \sum_{j=1}^J b_j (x-t_j)_+^{r-1}}\big\}  \colon a_k, b_k \in \RR   \Big\}
\$
Then we define the sieve $\Theta_n = \Theta_n^\alpha \times \Theta_n^\eta$, 
where $\Theta_n^\alpha = \{ x\mapsto \prod_{k=1}^d f_k(x_k)\colon f_k \in \mathrm{Spl}(J_n,0),k=1,\dots,d\}$ and $\Theta_n^\eta = \{ x\mapsto \prod_{k=1}^d f_k(x_k)\colon f_k \in \mathrm{Spl}(J_n),k=1,\dots,d\}$ for $J_n\to \infty$. 
\end{example}

We consider the classes of sufficiently smooth functions; 
for $p_1 = \lceil p \rceil -1$ and $p_2 = p-p_1$, we define 
\$
\Lambda_c^p = \Bigg\{ h \in C^{p_1}(\cX)\colon \sup_{\substack{x\in \cX \\ \sum_{l=1}^d \alpha_l <p_1}} |D^\alpha h(x)| + \sup_{\substack{x\notin x'\in \cX \\ \sum_{l=1}^d \beta_l = p_1}}  \frac{|D^\beta h(x) - D^\beta h(x')|}{\|x-x'\|^{p_2}}  \leq c   \Bigg\}
\$
To ensure non-negativeness, we also define the truncated function class 
$\Lambda_c^p(\cX, \epsilon) := \big\{ x\mapsto\max\{f(x),\epsilon \} \colon f\in \Lambda_c^p(\cX)\big\}$, 
obtained by thresholding $\Lambda_c^p(\cX)$ away from zero. 

For notational convenience, we denote the risk function
$\ell(\theta,x,y) = a  f^*\big( \frac{y + b }{- a} \big) + b + a \rho$ 
for $\theta = (a,b)$ as in Proposition~\ref{prop:dual_cond}. 
When there is no confusion, 
we equivalently use $\ell(\theta,x,y) = \ell((\alpha(x),\eta(x)),x,y)$ 
when $\theta = (\alpha,\eta)$ is a function. 
As preparation, 
we impose the following assumptions on 
the true optimizer 
and regularity conditions of the loss function. 

\begin{assumption}\label{assump:sieve}
Supopse $\cX = \prod_{k=1}^d \cX_d$ is the Catesian product of 
compact intervals, and $\theta^*\in \Theta = \Lambda_c^p(\cX,\epsilon) \times \Lambda_c^p(\cX)$ for some $c>0$. 
Suppose $\PP_{X\given T=1}$ has positive density on $\cX$. 
We assume the function $\EE[\ell((a,b),x,Y)\given X=x]$ is $\lambda$-strongly convex at $(a,b) = \theta^*(x)$ 
for all $x\in \cX$. Also, $|\ell(\theta,x,y) - \ell(\theta^* ,x,y)|\leq \bar\ell(x,y)\|\theta(x)-\theta^*(x)\|_2$ for $\|\theta(x) -\theta^*(x)\|_2<\epsilon$ 
for sufficiently small $\epsilon>0$, where $\|\cdot\|_2$ is the 
Euclidean norm, and 
$\sup_{x\in \cX}\EE[\bar\ell(x,Y)^2\given X=x,T=1]< M$ for 
some constant $M>0$. 
Furthermore, there exists a constant $C_1$ such that 
$\EE[\ell(\theta,X,Y(1)) - \ell(\theta^*,X,Y(1))\given T=1] \leq C_1 \|\theta - \theta^*\|_{L_2(\PP_{\cdot\given T=1})}^2$ 
when $\theta\in \Lambda_c^p(\cX)^2$ and 
$\|\theta - \theta^*\|_{L_2(\PP_{\cdot\given T=1})}$ is sufficiently small. 
\end{assumption}

We include a detailed discussion of 
Assumption~\ref{assump:sieve} 
in 
Appendix~\ref{app:subsec_discuss_sieve}, 
where we provide concrete examples and justifications for these conditions.
In Assumption~\ref{assump:sieve}, 
we assume the true optimizer is sufficiently smooth, 
so that function approximator can learn it well. 
It can be satisfied if the conditional distribution 
$\PP_{Y(1)\given X,T}$ is sufficiently ``smooth'' in $x$. 
We require the strong convexity of the 
conditional risk function at its minimizer $\theta^*(x)$; 
it is typically the case if $Y(1)$ is not 
deterministic given $X$.  
The stability condition at $\theta^*(x)$ 
can be satisfied if $Y$ is not heavy-tailed. 
We also assume that the population risk 
is
stable in terms of $L_2(\PP_{\cdot\given T=1})$ norm 
of $\theta$, which  can be satisfied if 
$\EE[\ell(\theta,x,Y)\given X=x,T=1]$ is 
smooth or have Lipschitz derivatives.  

Under the above regularity conditions, 
we obtain convergence rates of 
the empirical risk minimizers $(\hat\alpha^{(j)},\hat\eta^{(j)})$. 
The proof of the following theorem is in Appendix~\ref{app:subsec_sieve}.

\begin{theorem}\label{thm:sieve}
Suppose Assumption~\ref{assump:sieve} holds. 
We set $J_n = (\frac{\log n}{n})^{1/(2p+d)}$ for the sieve estimators 
in Examples~\ref{ex:poly} and~\ref{ex:spl}, 
and suppose $\hat\theta^{(j)}$ satisfies 
$\hat\EE_n \big[\ell(\hat\theta^{(j)},X,Y(1))\big] \leq \inf_{\theta\in \Theta_n} \hat\EE_n \big[\ell(\theta ,X,Y(1))\big] - O_P((\frac{\log n}{n})^{2p/(2p+d)})$. Then 
employing the function classes 
given in Examples~\ref{ex:poly} or~\ref{ex:spl}, 
we have 
$\|\hat\theta^{(j)} - \theta^* \|_{L_2(\PP_{\cdot \given T=1})} = O_P\big((\frac{\log n}{n})^{p/(2p+d)}\big)$ and 
$\|\hat\theta^{(j)} - \theta^* \|_{\infty } = O_P\big((\frac{\log n}{n})^{2p^2/(2p+d)^2}\big)$. 
\end{theorem}

The above theorem shows that 
under reasonable smoothness of the optimizer 
and regularity conditions on the loss function, 
the empirical risk minimizer $\hat\alpha^{(j)},\hat\eta^{(j)}$ 
converges to the truth at certain rates. 
Besides the examples and guarantees we provide, 
similar results might be obtained for other function 
classes like wavelets~\citep{daubechies1992ten}, 
and the conditions in Assumption~\ref{assump:sieve} 
might be weakened or modified to account for more generality. 
Such extension is beyond the scope of this work.

%% file: sec4.tex

\label{subsec:theory_cond}

In this section, we provide the theroetical guarantees 
for our procedure in Section~\ref{subsec:cond_method}. 
We first show the consistency of the estimators,  
with the double robustness and one-side validity results. 
We then present inferential guarantees: 
we achieve root-$n$ inference for $\mu_{1,0}^-$ 
under slower-than-parametric convergence 
rates of the nuisance component estimation; 
moreover,  
even when the empirical risk minimization 
is not consistent to the optimum, 
our inference procedure can still be valid. 
Finally, we 
show how to leverage our procedure
to construct bounds for treatment effects.

\subsection{Double consistency and one-side validity}

We first discuss the consistency of our estimator: 
we show that $\hat\mu_{1,0}^-$ from Algorithm~\ref{alg:est} 
is doubly robust to nuisance estimation, which 
is in a similar spirit as many results in causal inference and missing data. 
Even more interestingly,  
our estimator is robust to the ERM step: 
given that  either 
  $\hat{r}^{(j)}$ or $\hat{h}^{(j)}$ is consistent, 
it converges to the true bound if the ERM step is consistent;  
otherwise, our estimator converges to a conservative but still valid lower bound of $\mu_{1,0}^-$. 
We call this ``one-side validity''.  

We impose a mild assumption on the convergence of ERM step; 
note that we do not assume the convergence to the true minimizer $(\alpha^*,\eta^*)$.

\begin{assumption}\label{assump:erm_consist}
For each $j$, the empirical optimizer $(\hat\alpha^{(j)},\hat\eta^{(j)})$ 
converges in sup-norm to some $(\alpha^\diamond, \eta^\diamond)$ 
such that for all $x\in \cX$, 
$|\ell(\theta^\diamond(x),x,y) - \ell((a,b),x,y) |\leq M(x,y) \|\theta^\diamond(x) - (a,b)\|_2$ 
for all $\|\theta^\diamond(x) - (a,b)\|_2 \leq \epsilon$ for some constant $\epsilon>0$, 
and $\EE[M(x,Y(1))^2\given X=x,T=1]\leq M$ for some constant $M>0$. 
Also, $\hat{r}^{(j)}$ are uniformly bounded, 
and $\hat{H}^{(j)}$, $\hat{h}^{(j)}$ have uniformly bounded second moments 
almost surely. 
\end{assumption}

In Assumption~\ref{assump:erm_consist}, we additionally assume a mild regularity condition on the first-order expansion at the limit;
it is satisfied if the loss function $\ell$ is differentiable or locally Lipschitz. 
The second moment condition is also mild and standard. 
The following theorem shows the double robustness as well as one-side validity 
of our estimator, whose proof is in Appendix~\ref{app:subsec_double_consist}. 

\begin{theorem}
\label{thm:double_consist}
Suppose Assumption~\ref{assump:erm_consist} holds for some fixed $\theta^\diamond = (\alpha^\diamond, \eta^\diamond)$. 
Assume either (i) $\|\hat{r}^{(j)} - r_{1,0}\|_{L_2(\PP_{X\given T=1})} = o_P(1)$ or (ii)
$\|\hat{h}^{(j)} - \bar{h}^{(j)}\|_{L_2(\PP_{X\given T=1})}= o_P(1)$. Then the following holds as $n\to \infty$. 
\begin{itemize}
\item 
If $ \theta^\diamond = \theta^*$, i.e., the ERM step is consistent, then $\hat\mu_{1,0}^- = \mu_{1,0}^- + o_P(1)$;
\item 
otherwise, $\hat{\mu}_{1,0}^- = \mu_{1,0}^\diamond + o_P(1)$ for some constant $\mu_{1,0}^\diamond \leq \mu_{1,0}^-$. 
\end{itemize}
\end{theorem}

The above double robustness property 
generalizes previous results 
in observational studies 
without unmeasured confounding~\citep{robins1994estimation}. 
Notably, in our partial identification setting, 
$\hat{h}^{(j)}$ only needs to be consistent for 
the conditional expectation for a 
upstream estimator $\hat{H}^{(j)}$ that might not be 
consistent for its target. 
Even more interestingly, 
we allow for the inconsistency of the ERM step and 
still obtain a valid lower bound for $\mu_{1,0}$. 
Similar one-side validity 
has been documented by a recent work of~\citet{dorn2021doubly}, 
where they work under the marginal sensitivity model of~\citet{tan2006distributional} 
and develop such property 
based on an exact characterization of the worst-case scenario. 
However, in our setting, the one-side validity 
is a relatively straightforward consequence of duality. 
It would be interesting to find connections between our results; 
for example, whether their result can also be implied by the duality.

\subsection{Wald-type inference for $\mu_{1,0}^-$}
\label{subsec:cond_infer}

We now turn to inferential guarantees. 
We show that our procedure yields valid Wald-type inference 
under slow convergence rates of nuisance estimations. 
We begin with some regularity conditions 
on the risk function.

\begin{assumption}\label{assump:cond_regularity} 
Let $\theta^* = (\alpha^*,\eta^*)$ be the minimizer in~\eqref{eq:dual_cond}. 
Suppose 
$\EE[\nabla_{a,b}\ell(a,b,x,Y(1))\given X=x,T=1] = \nabla_{a,b}\EE[\ell(a,b,x,Y(1))\given X=x,T=1] = 0$ at $(a,b) = (\alpha^*(x),\eta^*(x))$ 
for $\PP_{X\given T=1}$-almost all $x$. 
Suppose 
$\big|\ell(a ,b,x,y) - \ell(\theta^*,x,y) - \nabla_{a,b} \ell(\alpha^*(x),\eta^*(x),x,y)[\alpha^*(x)-a,\eta^*(x)-b]\big|\leq M(x,y) \|(\alpha^*(x)-a,\eta^*(x)-b)\|_2^2$ for some $(a,b)$ in some neighborhood of $(\alpha^*(x),\eta^*(x))$, where $\EE[M(x,Y(1))\given X=x,T=1]\leq M$ 
for some constant $M>0$ for all $x\in \cX$. 
Furthermore,  
$\|\ell(\theta,X,Y(1)) - \ell(\theta^*,X,Y(1))\|_{L_2(\PP_{\cdot\given T=1})}= O(\|\theta-\theta^*\|_{L_2(\PP_{X\given T=1})})$ for function $\theta$ 
in a small $L_2(\PP_{\cdot\given T=1})$-neighborhood of $\theta^*$.
\end{assumption}

In Assumption~\ref{assump:cond_regularity}, 
we require the risk function to be differentiable 
and admits a Taylor expansion near some optimizer, 
as well as a regularity condition 
on the exchangeability of differentiation 
and conditional expectation. 
These are mild conditions that are 
commonly adopted in the literature~\citep{van2000asymptotic}. 
The risk function is assumed to be stable, so that 
plugging in estimators of $\alpha^*,\eta^*$ won't 
cause large errors, which is also a mild condition 
that can be satisfied under a first-order Taylor expansion condition. 

We assume the following convergence rates, where we assume the ERM step 
is $o_P({n^{-1/4}})$ consistent, and 
the nuisance estimation error  of $\hat{r}^{(j)}$  and $\hat{h}^{(j)}$ 
has a product of order $o_P(n^{-1/2})$. 

\begin{assumption}\label{assump:cond_convergence}
Suppose for each $j$, 
$\|\hat{r}^{(j)} - r_{1,0}\|_{L_2(\PP_{X\given T=1})} \cdot \|\hat{h}^{(j)} - \bar{h}^{(j)}\|_{L_2(\PP_{X\given T=1})} = o_P(n^{-1/2})$, 
and 
$\|(\hat\alpha^{(j)} - \alpha^*,\hat\eta^{(j)}- \eta^* )\|_{L_2(\PP_{X\given T=1})}  
= o_P(n^{-1/4})$ for some optimizer $(\alpha^*(x),\eta^*(x))$ 
of~\eqref{eq:dual_cond} satisfying Assumption~\ref{assump:cond_regularity}. 
\end{assumption}

In Assumption~\ref{assump:cond_convergence}, 
the rate of $\hat{r}^{(j)}$ depends on the 
estimation of $e(x) =\PP(T=1\given X=x)$, 
a standard regression problem.  
The estimation of $\hat{h}^{(j)}$ is also 
a regression problem viewing $\hat{H}^{(j)}$ as fixed. 
Convergence rate guarantees for such conditional mean 
estimation problems are well-established in the 
literature~\citep{stone1982optimal,mallat1999wavelet,pagan1999nonparametric,shen1994convergence,wasserman2006all,simonoff2012smoothing}. 
The estimation of $(\hat\alpha^{(j)},\hat\eta^{(j)})$ 
has been 
discussed in Section~\ref{subsec:cond_method}.

Under the above two assumptions, we show that 
our estimator is asymptotically normal 
and the estimation error of nuisance component is negligible. 
The proof of Theorem~\ref{thm:cond} is deferred to Appendix~\ref{app:subsec_thm_cond}. 

\begin{theorem}\label{thm:cond}
Suppose Assumptions~\ref{assump:cond_regularity} and~\ref{assump:cond_convergence} hold. 
Then $\sqrt{n}(\hat{\mu}_{1,0}^- - \mu_{1,0}^-)\rightsquigarrow N(0,\Var(\phi_{1,-}(X,Y,T)))$, 
where  
\$ 
\phi_{1,-}(X_i,Y_i,T_i) = \frac{T_i}{p_1}  r_{1,0} (X_i ) \big[  {H}(X_i,Y_i(1)) -  {h}(X_i ) \big] + \frac{1-T_i}{p_0}  h(X_i).
\$
Here $p_1 = \PP(T=1)=1-p_0$, 
and we define  
$H(x,y) = \ell(\theta^*,x,y)$, $h(x) = \EE\big[H(X,Y(1))\biggiven X=x,T=1\big]$. All the 
expectations (variances) are induced by the observed distribution. 
Furthermore, define  
\$
\hat\sigma^2 = \frac{1}{\hat p_1} \bigg( \frac{1}{n_1}\sum_{i\in \cI_1} d_{1,i}^2 - \Big(\frac{1}{n_1}\sum_{i\in \cI_1} d_{1,i} \Big)^2 \bigg) + 
\frac{1}{\hat p_0} \bigg( \frac{1}{n_0}\sum_{i\in \cI_0} d_{0,i}^2 - \Big(\frac{1}{n_0}\sum_{i\in \cI_0} d_{0,i} \Big)^2 \bigg)
\$
where $\hat{p}_1 = |\cI_1|/n$, $\hat p_0 = |\cI_0|/n$, 
$
d_{1,i}= \hat{r}^{(j[i])}(X_i)\big(\hat{H}^{(j[i])}(X_i,Y_i) - \hat{h}^{(j[i])}(X_i)),
$
$d_{0,i} =  \hat{h}^{(j[i])}(X_i)$, 
and $j[i]\in\{1,2,3\}$ is the fold that sample $i$ lies in. 
Then $\sqrt{n}(\hat\nu_{1,0}^- - \nu_{1,0}^-)/\hat\sigma \rightsquigarrow N(0,1)$.
\end{theorem}

Similar results can be obtained for $\mu_{1,0}^+$, 
if we simply flip the sign of $Y(1)$ and 
flip back after running the same procedure. 
The above procedure can also be generalized 
to the inference of $\mu_{0,1}^\pm$; 
the simplest way might be just switching the two groups.  
We summarize these results in Appendix~\ref{app:subsec_all_cond_bounds} 
for completeness.

\subsection{Robustness to misspecification of ERM}

Our inferential guarantee in Theorem~\ref{thm:cond} relies 
on consistency of both the nonparametric regression and the ERM steps. 
While these conditions are relatively mild, 
in this part, we take a step further and 
note that our estimator is in particular robust to the ERM step.  

The following theorem shows that 
even though our empirical risk minimizers  
 $\hat\alpha^{(j)}$ and $\hat\eta^{(j)}$ 
converge to something else, 
our procedure still provide valid, albeit more conservative, inference on 
the lower bound of $\EE[Y(1)\given T=0]$. 
The proof of Theorem~\ref{thm:cond_robust} 
is in Appendix~\ref{app:subsec_cond_robust}.

\begin{theorem}\label{thm:cond_robust}
    Suppose Assumptions~\ref{assump:cond_regularity} 
    and~\ref{assump:cond_convergence} with $(\alpha^*,\eta^*)$ 
    replaced by some fixed $\theta^\diamond := (\alpha^\diamond, \eta^\diamond)$, 
    and the first  condition of Assumption~\ref{assump:cond_regularity} 
    is replaced by the local one:  
    $ 
    \EE\big[ r(X)  \nabla_{a,b}\{\ell(\alpha^\diamond(X),\eta^\diamond(X),X,Y(1))[\alpha^\diamond(x)-\alpha(X),\eta^\diamond(X)-\eta(X)]\}\biggiven  T=1\big] = 0$ for any  $(\alpha,\eta)\in \Theta_n$ in a small $\|\cdot\|_{\infty}$-neighborhood of $\theta^\diamond$. 
    We additionally assume $\|\hat\theta - \theta^\diamond\|_\infty = o_P(1)$. 
Then $\sqrt{n}(\hat\mu_{1,0}^- - \mu_{1,0}^\diamond)\rightsquigarrow N(0,\Var(\phi_{1,-}^{\diamond}(X,Y,T)))$, 
where $\mu_{1,0}^\diamond \leq \mu_{1,0}^-$, and 
\$ 
\phi_{1,-}^\diamond(X_i,Y_i,T_i) = \frac{T_i}{p_1}  r_{1,0} (X_i ) \big[  {H}^\diamond(X_i,Y_i(1)) -  {h}^\diamond(X_i ) \big] + \frac{1-T_i}{p_0}  h^\diamond(X_i).
\$
Here we define 
$H^\diamond(x,y) = \ell(\theta^\diamond,x,y)$ and $h^\diamond(x) = \EE\big[H^\diamond(X,Y(1))\biggiven X=x,T=1\big]$. 
Furthermore, we have $\sqrt{n}(\hat\mu_{1,0}^- - \mu_{1,0}^\diamond)/\hat\sigma \rightsquigarrow N(0, 1)$ for the variance estimator $\hat\sigma^2$ defined in Theorem~\ref{thm:cond}.
\end{theorem}

In theorem~\ref{thm:cond_robust}, 
we only require 
the convergence of $(\hat\alpha^{(j)},\hat\eta^{(j)})$ 
in $L_2(\PP_{\cdot\given T=1})$-norm  
any pair of fixed functions. 
This might happen, for example, if the function class we employ 
does not approximate $(\alpha^*,\eta^*)$ very well, 
but our estimators still converge to a fixed in-class risk minimizer.  
In this case, our 
estimator converges to a conservative lower bound of the counterfactual mean 
and still yields valid inference. 

In parallel to the mean-zero gradient 
property of $\theta^*$, 
we assume a local first-order condition 
for $\theta^\diamond$ 
restricted to $\Theta_n$, which is crucial for the double robustness to the estimation error. 
This condition is satisfied  
as long as $\theta^\diamond$ is the population risk minimizer (with weight $r(X)$) 
among $\Theta_n$. 
To obtain an estimator that converges to $\theta^\diamond$, 
we might slightly change the procedure: 
fit  $\hat{r}^{(j)}(x)$ on one fold and 
and run the ERM with the fitted $\hat{r}^{(j)}(x)$ on a new fold. 
The convergence of the empirical risk minimizer 
can be satisfied if $\Theta_n$ 
is not too complex and $\hat{r}^{(j)}$ is consistent with a slow rate. 

We also note that, as implied by Theorem~\ref{thm:cond_robust}, 
plugging in any fixed function into our procedure without ERM (or equivalently, setting $\Theta_n = \{\theta\}$ for some fixed $\theta$ that satisfy the regularity conditions)
also yields a valid lower bound. However, this is uninteresting 
as it may be way too conservative.

\subsection{Inference for treatment effects}
\label{subsec:cond_ate}

With the above estimator for counterfactual means in place, 
we briefly discuss the construction of 
confidence intervals for treatment effects. 
Let us first start with ATT/ATC. 
Following the preceding example of $\mu_{1,0}^-$, 
in view of~\eqref{eq:bound_ate}, we can construct an 
estimator for the lower bound of ATC, defined as 
\$
\hat\tau_{\textrm{ATC}}^- := \hat\mu_{1,0}^- - \frac{1}{n_0}\sum_{i\in \cI_0}Y_i,
\$
where $\hat\mu_{1,0}^-$ is constructed as in Section~\ref{subsec:cond_method}. 
Theorem~\ref{thm:cond} directly implies the 
following result of double robustness 
and asymptotic normality for $\hat\tau_{\textrm{ATC}}^-$,  
and the proof is omitted for brevity. 

\begin{corollary}
\label{cor:atc_lower}
Under the same conditions of Theorem~\ref{thm:cond}, 
$\sqrt{n}(\hat\tau_{\atc}^- - \tau_{\atc}^-)\rightsquigarrow N(0,\Var(\phi_{\atc}^-(X_i,Y_i,T_i)))$, where $\tau_{\atc}^- = \mu_{1,0}^{-} - \EE[Y(0)\given T=0]$ is a lower bound for ATC under the $(f,\rho)$-selection condition, and 
\$
\phi_{\atc}^-(X_i,Y_i,T_i) = \frac{T_i}{p_1} r_{1,0}(X_i)\big[ H(X_i,Y_i(1)) - h(X_i) \big] + \frac{1-T_i}{p_0} \big( h(X_i) + Y_i(0)\big). 
\$
\end{corollary}

Similar to Theorem~\ref{thm:cond}, 
a consistent estimator $\hat\sigma_{\atc,-}^2$ can also 
be constructed for $\Var(\phi_{\atc}^-(X_i,Y_i,T_i))$, 
enabling Wald-type inference. 
Based on the results of Section~\ref{subsec:theory_cond}, 
we can similarly construct other bounds 
for ATT and ATC and combine them to obtain bounds on ATE. 
For example, let 
$\hat\mu_{0,1}^+$ estimate an upper bound on $\EE[Y(0)\given T=1]$ 
with influence function $\phi_{0,1}^+$ 
(see Appendix~\ref{app:subsec_all_cond_bounds} for details).
We may construct  
\$
\hat\tau_{\att}^- := \frac{1}{n_1} \sum_{i\in \cI_1} Y_i - \hat\mu_{0,1}^+, 
\quad \textrm{and} \quad 
\hat\tau_{\ate}^- := \hat{p_1} \cdot \hat\tau_{\att}^- + \hat p_0 \cdot \tau_{\atc}^-.
\$
Then $\sqrt{n}(\hat\tau_{\ate}^- - \tau_{\ate}^-)\rightsquigarrow N(0,\Var(\phi_{\ate}^-(X_i,Y_i,T_i)))$, where $\tau_{\ate}^-$ 
is a lower bound for ATE under the $(f,\rho)$-selection condition, and 
the influence functions are 
$\phi_{\ate}^- = p_1 \phi_{\att}^- + p_0 \phi_{\atc}^-$, and 
$\phi_{\att}^- = T_iY_i/p_1 - \phi_{0,1}^+ $.

%% file: sec5.tex

We illustrate the performance of our procedure on simulated datasets. 
We focus on the estimation of the counterfactual mean $\EE[Y(1)\given T=0]$ 
given confounded observational data and take $f(t)=t\log t$.

\subsection{Simulation setting}
We fix the sample size at $n=15000$ and the covariate dimension at $p=4$. 
To generate the confounded dataset, 
setting $U=Y(1)$, 
we fix the observed distribution of $\PP_{X,Y(1)\given T=1}$ and $\PP(T=1)$, 
and vary the counterfactual distribution $\PP_{X,Y(1)\given T=0}$, 
so that $\ORR(x,u) := \frac{\PP(T=0\given X=x,U=u)\PP(T=1\given X=x)}{\PP(T=1\given X=x,U=u)\PP(T=0\given X=x)}$ 
satisfies $(f,\rho)$-selection condition for a sequence of $\rho>0$. 
To be specific, we generate the covariates and treatment assignments with 
\$
X\sim \textrm{Unif}[0,1]^p,\quad T\given X \sim \textrm{Bern}(e(X)), 
\$
where we set the observed propensity score as $e(x) = \textrm{logit}(\gamma^\top x)$ 
for $\gamma = (-0.531, 0.126, -0.312, 0.018)^\top$.  
Finally, given $\delta\in\RR$, we generate the potential outcomes via 
\$
Y(1)  =  X^\top \beta_1  - \delta \cdot  (1-T) \sigma(X) + \varepsilon \cdot \sigma(X), \\
Y(0)  =  X^\top \beta_0  - \delta \cdot  (1-T) \sigma(X) + \varepsilon \cdot \sigma(X),
\$
where $ \varepsilon \iid N(0,1)$, and we set 
 $\beta_1 = (0.531, 1.126, -0.312, 0.671)^\top$, 
$\beta_0 = (-0.531, -0.126, -0.312, 0.671)^\top$ and 
$\sigma^2(x) = 1+1.25x_1^2$. 
Put it another way,  
the observations of $Y(1)$ in the treated group follow 
$Y(1)\given X=x,T=1\sim N(x^\top \beta_1, \sigma^2(x))$, 
while 
$Y(1)\given X=x,T=0\sim N(x^\top \beta_1 - \delta \cdot\sigma(x), \sigma^2(x))$.

In this setting, the confounder is entirely driven by $U:= Y(1)$.  
The odds ratio is 
\$
\ORR(x,u) = \exp\Big(  - \frac{\delta (u-x^\top \beta_1) + \delta^2}{2\sigma^2(x)}   \Big),
\$
and we obtain an upper bound for the $f$-divergence as $\rho = \delta^2/2$. 
The same bound can be obtained for the other odds ratio of the control group. 
The observed dataset is thus $\{(X_i,Y_i,T_i)\}_{i=1}^n$, where $Y_i = Y_i(T_i)$. 
Intuitively, $\delta$ drives the direction and magnitude of confounding: 
when $\delta>0$, larger values of $Y(1)$ has larger probability of getting treated even conditional on $X$; 
as a result, the observed $Y(1)$ in the treated group 
is actually shifted to larger values, leading to overestimate of treatment effects 
if confounding is not accounted for. The larger $\delta$ is, 
the more severe the impact of confounding is. 
On the other hand, when $\delta<0$, inference under the strong ignorability assumption 
tends to underestimate the treatment efects. 
In this setting, although we anticipate OR$(X,U)$ to be controlled overall, 
it does not admit a uniform upper bound; we plot several quantiles of 
OR$(X,U)$ in the treated group in Figure~\ref{fig:quantiles}. 

\begin{figure}[H]
\centering 
\includegraphics[width=3in]{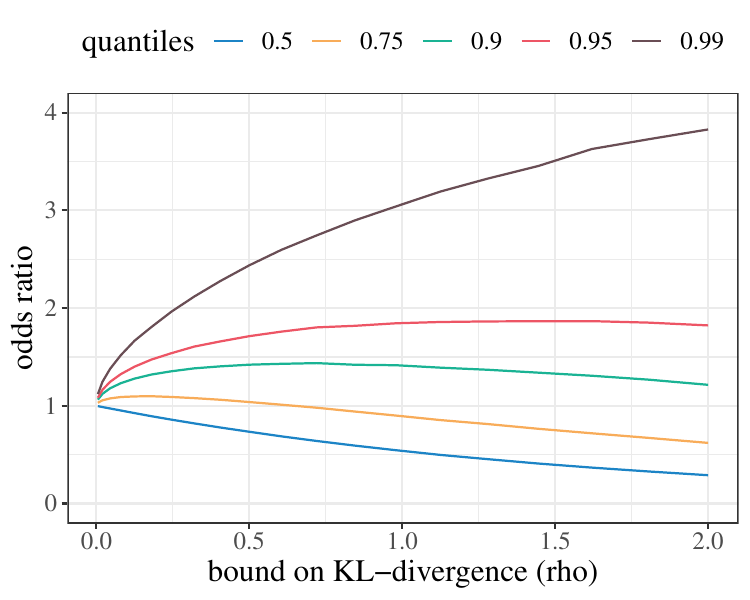}
\caption{Quantiles of OR$(X,U)$ in the treated group for a sequence of $\delta$ and $\rho$ (the $x$-axis).}
\label{fig:quantiles}
\end{figure}

We apply Algorithm~\ref{alg:est} to obtain bounds  and confidence intervals of $\EE[Y(1)\given T=1]$. 
The detailed implementation is as follows: 
the regression algorithm we use for both $\hat{r}$ and $\hat{h}$ 
is Random Forest Regressor from \texttt{scikit-learn} Python library~\citep{scikit-learn}; 
we use cubic spline in Example~\ref{ex:spl} to approximate $\alpha^*,\eta^*$, 
where we threshold at $\epsilon = 0.001$ to guarantee positiveness of $\alpha^*$ 
(yet our estimates turn out to be strictly larger than this threshold). 
We employ the Nelder-Mead optimizer implemented in \texttt{SciPy} Python library~\citep{2020SciPy-NMeth}
to optimize the coefficients in the spline approximation.

\subsection{Sensitivity analysis with one dataset}

We first illustrate the estimators and confidence intervals we obtain 
under a fixed confounded data generating process. 
To be specific, we fix $\delta = 0.5$ (hence $\rho = 0.125$) to generate the data, 
and apply our procedure to the fixed dataset for a series of $\rho \in \{0.05, 0.1,\dots, 0.95, 1.0\}$. 
We obtain $0.975$-lower confidence bound (LCB) for the lower bound of ATC and $0.975$-upper confidence bounds (UCB) for 
the upper bound of ATC (i.e., the bounds in~\eqref{eq:bound_atc}), 
which together form a $0.95$-CI for ATC under a hypothesized confounding level $\rho$. 
The results are plotted in Figure~\ref{fig:fix}. 

Without accounting for confounding, reweighting on the covariates 
tend to overestimate the ATC (indicated by the estimators for small $\rho$). 
The LCB crosses the ground truth at $\hat \rho=0.1$; 
this can be viewed as 
a lower confidence bound for 
the true confounding level $\rho=0.125$ (we elaborate on this in the discussion when the ground 
truth is zero). 
Finally, the LCB hits zero at $\hat\rho_0 = 0.65$; 
we can thus conclude with $0.95$-confidence  that 
ATC is non-negative as long as the true 
confounding level does not exceed $\hat\rho_0$.

\begin{figure}[h]
\centering 
\includegraphics[width=6in]{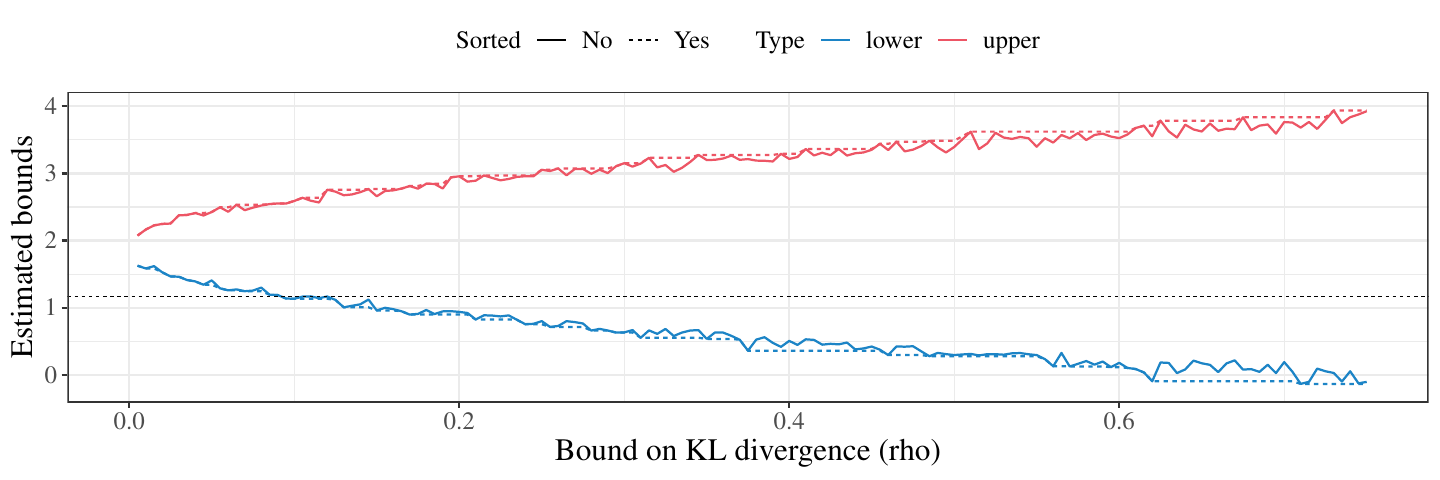}
\caption{The $0.975$-LCB for the lower bound of ATC (blue) and $0.975$-UCB for 
the upper bound of ATC (red), obtained from one run of our procedure on one dataset. 
Solid lines are the original estimators from our procedure, while dashed lines are sorted to ensure they are monotone in $\rho$. The black dashed line is the actual ATC. 
}
\label{fig:fix}
\end{figure}

\subsection{Validity and sharpness}

To show the validity and sharpness of our procedure, 
we first vary $\delta \in \{0.1,0.2,\dots,1.5\}$ in our data-generating process, 
and apply our procedure with the correct level $\rho = \delta^2/2$. 
Feeding the data into Algorithm~\ref{alg:est} yields the estimator $\hat \mu_{1,0}^-$ 
for the lower bound on $\EE[Y(1)\given T=1]$; 
changing the observations to $Y(1)\leftarrow -Y(1)$, 
the negative of the output of Algorithm~\ref{alg:est}, 
denoted as $\hat\mu_{1,0}^+$, is an estimator for the upper bound $\mu_{1,0}^+$. 
Based on the corresponding variance estimators $\hat\sigma_{1,0,\pm}$, 
we construct the confidence interval for $\EE[Y(1)\given T=0]$ as 
$\text{CI}_{\text{mean}} := [\hat\mu_{1,0}^- + z_{0.025} \hat\sigma_{1,0,-}/\sqrt{n}, 
\hat\mu_{1,0}^+ + z_{0.975} \hat\sigma_{1,0,+}/\sqrt{n}]$; 
the confidence interval for $\mu_{1,0}^-$ is constructed as 
$\text{CI}_{\text{lower}} :=[\hat\mu_{1,0}^- + z_{0.025} \hat\sigma_{1,0,-}/\sqrt{n}, 
\hat\mu_{1,0}^- + z_{0.975} \hat\sigma_{1,0,-}/\sqrt{n}]$, 
and similarly 
$\text{CI}_{\text{upper}} :=[\hat\mu_{1,0}^+ + z_{0.025} \hat\sigma_{1,0,+}/\sqrt{n}, 
\hat\mu_{1,0}^+ + z_{0.975} \hat\sigma_{1,0,+}/\sqrt{n}]$ 
 for $\mu_{1,0}^+$. 
 
To obtain the ground truth of $\mu_{1,0}^\pm$ at each $\delta$, 
we evaluate the bounds on $\EE[Y(1)\given X=x,T=0]$ for each $x$  
by optimizing with
a huge amount of samples from $\PP_{Y(1)\given X=x,T=1}$;\footnote{This is feasible because 
in our setting, $\PP_{Y(1)\given X=x,T=1}$ 
is normal distribution, and the 
target bounds are shift-invariant; we only need to 
evaluate the bounds for all values of $\rho$ and a fine grid of $\sigma(x)$.}
we then 
marginalize over $X\given T=0$ to obtain an estimator the ground truth of $\mu_{1,0}^\pm$. 
For each $\delta$, this procedure is repeated and 
averaged over many runs to further reduce the random error.

The estimators for bounds of counterfactuals over $N=500$ runs for each $\rho$ 
are plotted in Figure~\ref{fig:est} (they are evenly spaced on the $x$-axis). 
The simulation results show the sharpness and accuracy of our estimators:  
they are quite close to the ground truth, especially for small values of $\rho$; 
they get a bit conservative and have a larger variance when $\rho$ is as large as $1$. 
Interestingly, there are also a few outliers when $\rho$ is very small, 
and the estimators seem to be the most stable for an medium scale of $\rho$ (around 0.18 to 0.5). 
The actual value of $\EE[Y(1)\given T=0]$ in our design, represented by the red triangles, 
are very close to the lower solid line, the ground truth of $\mu_{1,0}^-$; 
this means our simulation design is close to the worst case. 

\begin{figure}[h]
\centering 
\includegraphics[width=5in]{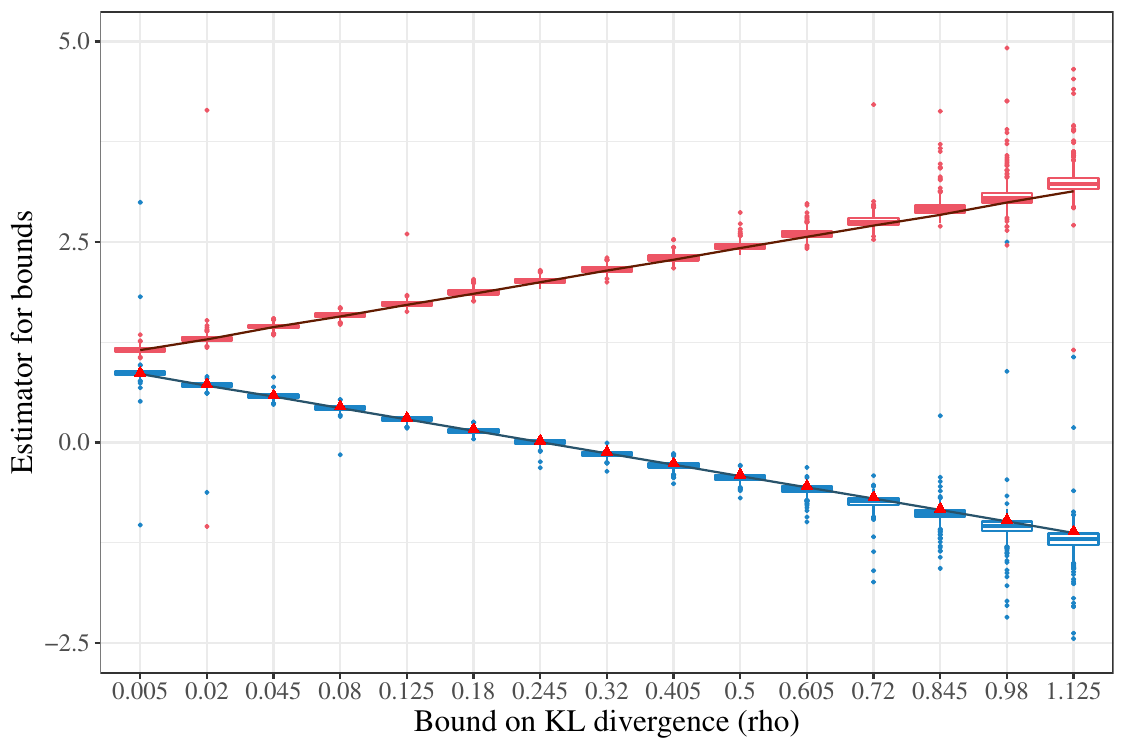}
\caption{Boxplots for $\hat\mu_{1,0}^+$ (red ones) and $\hat\mu_{1,0}^-$ (blue ones) over $N=500$ replicates 
with each value of $\rho$. 
The solid lines are the ground truths of $\mu_{1,0}^+$ and $\mu_{1,0}^-$. 
The red triangles represent the actual value of $\EE[Y(1)\given T=0]$ in our simulation setting. 
}
\label{fig:est}
\end{figure}

To further validate our inference procedure, 
we compute the empirical coverage of 
$\text{CI}_{\text{lower}}$ and $\text{CI}_{\text{upper}}$ for $\mu_{1,0}^\pm$ over $N=500$ runs. 
We also compute the ground truth of $\EE[Y(1)\given T=0]$ under our design as a baseline, 
and compute the empirical coverage of $\text{CI}_{\text{mean}}$. 
They are plotted in Figure~\ref{fig:cover}. 
Our empirical coverage is close to the nominal level $0.95$ in almost all settings, 
showing the validity of our inference procedure. 

\begin{figure}[H]
\centering 
\includegraphics[width=6in]{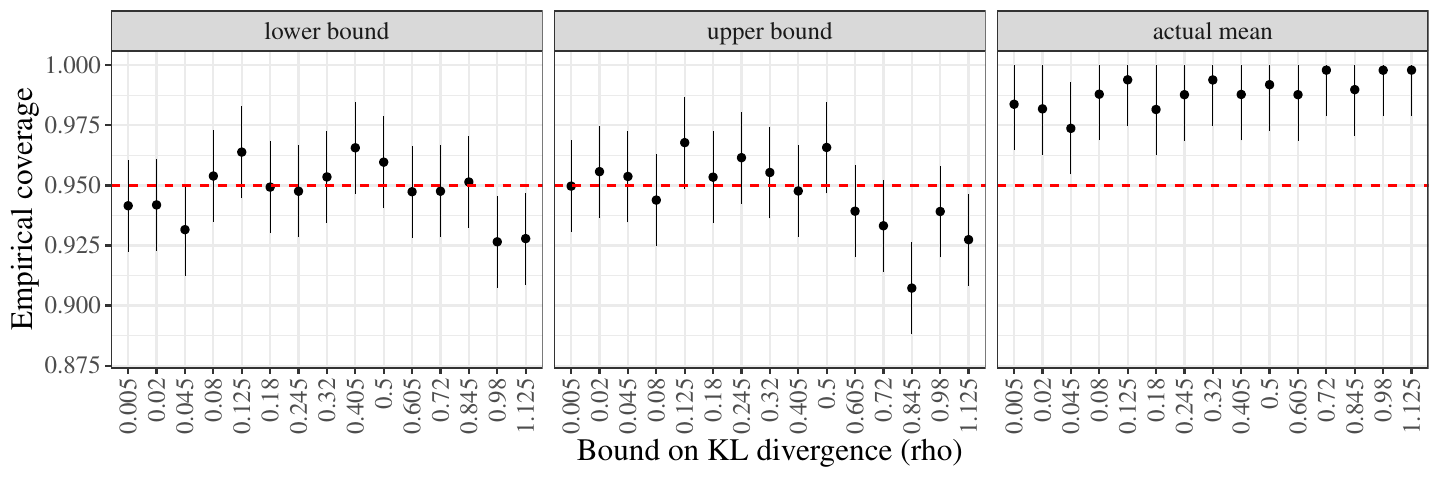}
\caption{Empirical coverage for $\mu_{1,0}^-$ (left), $\mu_{1,0}^+$ (middle), and $\EE[Y(1)\given T=0]$ (right). 
The short vertical segments are the C.I.s computed with $N=500$ replicates. 
The red dashed line is the nominal level $0.95$. 
}
\label{fig:cover}
\end{figure}

Figure~\ref{fig:1side_cover} 
plots the empirical coverage of one-sided C.I.s for $\mu_{1,0}^\pm$, defined as 
$\text{CI}_{\text{lower}}^{\textrm{one-side}} :=[\hat\mu_{1,0}^- + z_{0.05} \hat\sigma_{1,0,-}/\sqrt{n}, 
+\infty)$  
and  
$\text{CI}_{\text{upper}}^{\textrm{one-side}} :=(-\infty,  
\hat\mu_{1,0}^+ + z_{0.95} \hat\sigma_{1,0,+}/\sqrt{n}]$.  
Our theory shows that even though the ERM is off, these C.I.s 
still have valid asymptotic coverage; 
such robustness is also supported by empirical evidence. 

\begin{figure}[H]
\centering 
\includegraphics[width=4in]{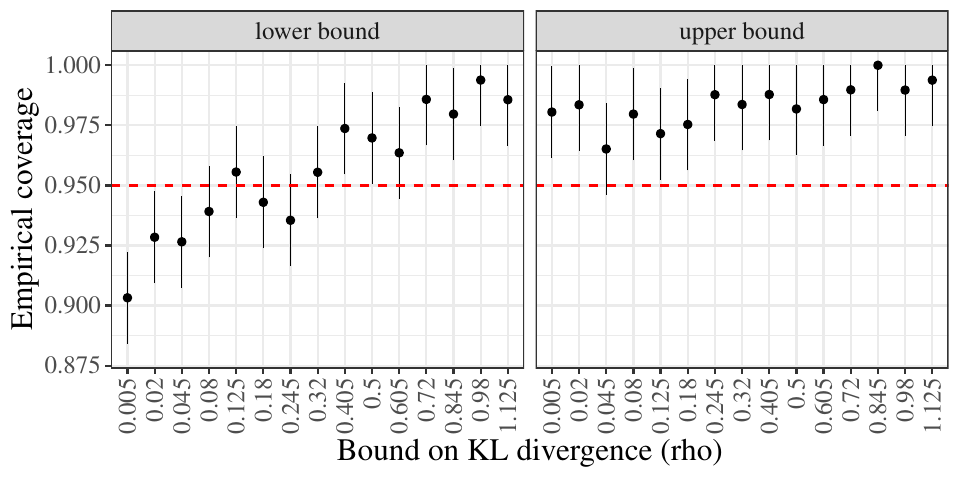}
\caption{Empirical coverage of one-sided C.I.s for $\mu_{1,0}^-$ (left), and $\mu_{1,0}^+$ (right).
The short vertical segments are the C.I.s computed with $N=500$ replicates. 
The red dashed line is the nominal level $0.95$. 
}
\label{fig:1side_cover}
\end{figure}

%% file: discussion.tex
%
\label{sec:discussion}
In this work, we propose a new sensitivity model
based on the $f$-divergence that characterizes 
the {\em average} effect of confounders on selection bias.
Under the $f$-sensitivity model, we offer a
scheme for the estimation and inference on
the counterfactual and the ATE. We close the paper by
a discussion on possible extensions.

\paragraph{Tightness.}
As mentioned before, the optimal value of~\eqref{eq:opt_cond}
is not necessarily the tightest lower bound for $\EE[Y(1) \given T = 0]$:
the sharp one 
under $(f,\rho)$-selection condition 
is given by 
\$
\inf \Big\{  \EE^{\sup}\big[Y(1)\biggiven T=0\big] \colon  {\PP^{\sup} \in \cQ_{1,0}} \Big\},
\$
where $\cQ_{1,0}$ is the identification set of all distributions that agree with the 
observed distribution and 
satisfy the $(f,\rho)$-selection condition.
The constraints in~\eqref{eq:opt_cond} define a superset of $\cQ_{1,0}$,
potentially leading to conservativeness.

Using the exact characterization of $\cQ_{1,0}$
provided in Proposition~\ref{prop:robust_set_general},
we can represent 
the sharp lower (resp.~upper) bound of $\EE[Y(1)\given T=0]$ 
under the $(f,\rho)$-selection condition as the 
optimal value of
\$
\mathop{\min (\text{resp.~}\max)}_{L(x) \textnormal{~measurable}}~
             &\E\big[Y(1) L(X) \biggiven T=1\big]\\
  \textnormal{s.t.}~&\E[L(x) \given X=x, T=1]= r_{1,0}(x) \\
                    &\E\big[f\big(L(x) /r_{1,0}(x)  \big)\biggiven X=x,T=1\big]\le \rho,\quad \mbox{for almost all }x. \\
                    &\E\big[r_{1,0}(x) f\big(r_{1,0}(x)/L(x)  \big)\biggiven X=x,T=1\big]\le \rho,\quad \mbox{for almost all }x.
\$ 
With the same argument, we can also
develop the optimization problems
for 
sharp bounds on the ATT and the ATC. 
Compared to the dual problems in 
Proposition~\ref{prop:dual_cond}, 
the additional constraints in 
the last line above 
leads to a dual form that is not 
as clean. Developing an efficient
algorithm that solves this tight bound 
remains an interesting
avenue for future research. 

\paragraph{Sensitivity analysis.}
In this paper, we have focused on conducting inference on 
the counterfactuals and treatment effects under the $(f,\rho)$-selection condition,
with a prescribed confounding parameter $\rho$. 
Based on this, we can make robust causal 
conclusions and conduct sensitivity analysis by inverting the confidence intervals 
 as follows. 
Suppose the goal is to detect if there is a nonzero 
ATE; we can consider a increasing sequence of $\rho$, and construct 
a level $1-\alpha$ confidence
interval $\hat C(\rho)$ for the ATE using the method introduced in this paper
at each value of $\rho$;
finally let $\hat \rho$ be the smallest $\rho$ such that $C(\rho)$ contains 
zero. We can interprete the results as 
either there is a nonzero ATE, or there is a confounder as large 
as $\hat{\rho}$ to explain away the observed treatment effects.

More rigorously, let $\rho^*$ denote the true confounding level
and suppose the constructed confidence intervals
$\hat C(\rho)$ are nested in $\rho$: for any $\rho_1 \le \rho_2$,
$\hat C(\rho_1) \subset \hat C(\rho_2)$.  
We then have 
\$
\limsup_{n\to \infty} \PP(\mbox{ATE } = 0, \rho^* < \hat \rho) \le 
\limsup_{n\to \infty} \PP(\mbox{ATE } \notin \hat C(\rho^*)) \le \alpha,
\$
if $\hat{C}(\rho^*)$ is an asymptotically valid confidence interval 
for the ATE. In words, when the ATE is indeed zero,
$\hat \rho$ is an asymptotic level-$(1-\alpha)$ confidence lower
bound for $\rho^*$. 
Similar to the case of~\cite{jin2021sensitivity}, 
here only point-wise validity is necessary, i.e., 
we only need our CIs to be asymptotically valid for each fixed ground truth of $\rho$. 
Finally, we note that the monotonicity of the confidence intervals 
is satisfied with a reasonable estimation procedure; 
one can also force the confidence intervals to be monotone 
by enlarging some of them to conform to those for smaller values of $\rho$, 
without hurting the asymptotic validity.

\paragraph{Implications for the conditional average treatment effect (CATE).}
The mothodology proposed in this paper 
also provides bounds on CATE under the $(f,\rho)$-selection condition. 
For example, the proof of Propositions~\ref{prop:convex_prob} and~\ref{prop:dual_cond} 
implies that a lower bound for $\EE[Y(1) \given X = x,T=1]$ is given by 
the optimal value of  
\$
\min_{L\geq 0 \textrm{ measurable}}~&\EE\big[Y(1) L \given X = x, T = 1\big]\\
s.t.~& \EE\big[L \given T=1, X=x\big],\\
& \EE\bigg[f\Big(\frac{L}{r_{1,0}(x)}\Big) \given X  =x, T=1\bigg] \le \rho. 
\$
The dual form of the above optimization problem is  
\#\label{eq:cate_dual}
\sup_{\alpha\ge 0, \eta \in \RR}~
-r_{1,0}(x) \cdot \EE\bigg[\alpha f^*\Big(-\frac{Y(1) + \eta}{\alpha}\Big)
 +\eta + \alpha\rho \Biggiven X = x,T=1 \bigg].
\#
Note that the optimizer $(\alpha^*(x),\eta^*(x))$ defined 
in~\eqref{eq:perx_opt} is exactly the optimizer of~\eqref{eq:cate_dual}.
In fact, $\hat\mu(x):= \hat{r}^{(j)}(x)\hat{h}^{(j)}(x)$ where 
$\hat{r}^{(j)},\hat{h}^{(j)}$ are defined in Algorithm~\ref{alg:est} 
is an estimator for the optimal objective in~\eqref{eq:cate_dual}. 
These quantities are repeatedly estimated on distinct folds of data 
as intermediate steps of our procedure. 
While such sample splitting does not compromise the efficiency of inference 
due to the final averaging step,  
how to efficiently estimate these CATE bound functions 
with statistical guarantees 
might call for distinct considerations from ours. 
We leave this for future investigation.

\paragraph{Marginal $(f,\rho)$ selection condition.} 
We might even relax the per-$x$ uniform bound on the $f$-divergence 
in Definition~\ref{def:cond}
to a marginal fashion, so that the selection bias is controlled 
averaged over both $U$ and $X$. More formally, we might consider 
the constraint that 
\$
\int f\Big(\frac{\PP(T=0 \given X=x,U)}{\PP(T=1\given X=x,U)} 
\frac{\PP(T = 1 \given X=x)}{\PP(T=0 \given X=x)}\Big) 
\diff \PP_{U,X \given T=1} \leq \rho. 
\$
In this setting, the odds ratio can be very large 
for a small proportion of $X\given T=1$, 
but still controlled in the average sense. 
This type of marginal $(f,\rho)$-selection model leads 
to a larger class of distributional shifts than 
the $(f,\rho)$-selection condition here, 
and a different optimization problem for bounds on counterfactual means. 
Following similar arguments here, 
we see that the dual formulation, parallel to Proposition~\ref{prop:dual_cond}, 
can still be viewed as a risk minimization problem; however,  
the risk function would involve the unknown $X$-shift $r_{1,0}$, 
which might make the estimation and inference more complicated. 
The estimation and inference under this marginal $f$-sensitivity model 
is an ongoing work.

%% file: appendix.tex

\section{Proofs for identification sets} 
\label{app:proof_bound_shift}

\subsection{Proof of Proposition~\ref{prop:robust_set}}
\label{app:subsec_idset}

Here we prove a stronger result that directly implies
Proposition~\ref{prop:robust_set}. The
following proposition is a tight characterization 
of the identification set induced by the $f$-sensitivity models.

\begin{proposition}\label{prop:robust_set_general}
Let $(X,Y(0),Y(1),U,T)\sim \PP^{\sup}$ be the true
unknown super-population over all random variables of interest. 
Let $\cP$ be the set of all distributions over $(X,Y(0),Y(1),U,T)$. 
Let $\PP^{\obs}_{X,Y,T}$ be the joint distribution 
of all observable random variables $(X,Y,T)$.  
Let $t\in\{0,1\}$. 
Define $\cQ_{t,1-t}$ as the set of all 
counterfactual distributions that 
agrees with the observables and satisfies the $(f,\rho)$ selection condition, 
i.e.,
\$
\cQ_{t,1-t} = \big\{ \PP^{ }_{X,Y(t)\given T=1-t}  \colon \PP\in \cP,~ \PP^{ }_{X,Y,T} = \PP^{\obs}_{X,Y,T},~ \PP^{ }\,\text{satisfies Definition~\ref{def:cond}} \,  \big\}.
\$
Then  
$\PP^{\sup}_{X,Y(t)\given T=1-t} \in \cQ_{t,1-t}$, and 
\$
\cQ_{t,1-t} =  \Bigg\{ \QQ \colon  \frac{\ud \QQ_X}{\ud \PP_{X\given T=t}^{\obs}}(x) = r_{t,1-t}(x),
  ~&D_f\big( \QQ_{Y\given X=x} ~\big\|~  \PP_{Y\given X=x,T=t}^{\obs} \big) \leq \rho,
  ~\text{for }\PP^\obs_{X\given T=t}\text{-almost all }x ,\\
  ~&\quad D_f\big( \QQ_{Y^{\obs}\given X=x, T=t} ~\big\|~  \QQ_{Y\given X=x} \big) \leq \rho,
  ~\text{for }\QQ_{X}\text{-almost all }x  \Bigg\} ,
\$
where $r_{1,0}(x) = \frac{(1-e(x))p_1}{e(x)(1-p_1)}$, $r_{0,1}(x) = \frac{e(x)(1-p_1)}{(1-e(x))p_1}$, and 
$e(x) = \PP^\obs(T=1\given X=x)$, $p_1=\PP^\obs(T=1)$. 
\end{proposition}

\begin{proof}[Proof of Proposition~\ref{prop:robust_set_general}]
Fix $t = 1$. For any $\PP_{X,Y(1)\given T=0} \in \cQ_{1,0}$,
since $\PP_{X,Y,T} = \PP^{\rm obs}_{X,Y,T}$, 
\$
\frac{\ud \PP_{X \given T=0}}
{\ud \PP^{\rm obs}_{X \given T = 1}}
= 
\frac{\ud \PP_{X \given T=0}}
{\ud \PP_{X \given T = 1}}
= r_{1,0}(x).
\$
By Lemma~\ref{lem:bound_shift_single},
$
D_f\big(\PP_{X,Y(1) \given T = 0} ~\big\|~
\PP_{X,Y(1) \given T = 1}\big) \le \rho.
$
On the other hand,  
\$
D_f\big(\PP_{Y(1) \given X, T=1} 
~\|~ \PP_{Y(1) \given X, T=0}\big)
\le & 
D_f\big(\PP_{Y(1),U \given X, T=1} 
~\|~ \PP_{Y(1),U \given X, T=0}\big)\\
= & 
\E_{\PP_{Y(1),U \given X,T=0}}\bigg[f\Big(\frac{\ud \PP_{Y(1),U \given X,T = 1}}
{\ud \PP_{Y(1), U \given X, T=0}}\Big)\bigg]\\
= & 
\E_{U \given X,T=0}\bigg[f\Big(\frac{\PP(T = 1\given X, U)}{\PP(T = 0\given X,U)}
\cdot \frac{\PP(T = 0\given X)}{\PP(T=1 \given X)}\Big)\bigg]\le \rho,
\$
where the last inequality is due to the
$(f,\rho)$-selection condition.
Combining the above, we establish the 
``$\subset$'' direction.
It remains to prove the reverse. We show the proof
for the case of $t=1$ here, and the $t=0$ case follows from 
similar arguments. 

Given any $Q \in \cQ_{1,0}$, we aim to find a distribution 
$\PP^{\rm sup}$ over $(X,Y(0),Y(1),U,T)$ such that 
\begin{itemize}
\item $(Y(1),Y(0)) \indep T \given X,U$
\item $\PP^{\rm sup}$ is compatible with $\PP^{\rm obs}_{X,T,Y}$;
\item $\PP^{\rm sup}$ satisfies the $(f,\rho)$-selection condition;
\item $\PP^{\rm sup}_{X,Y(1)\given T=0}(x,y) = Q(x,y)$.
\end{itemize}
To construct $\PP^{\rm sup}$, we first set
$\PP^{\rm sup}_{X,T} = \PP^{\rm obs}_{X,T}$.
Then we specify the distribution of $Y(1) \given X,T$ via
\$
\PP^{\rm sup}_{Y(1) \given X, T = 1} = \PP^{\rm obs}_{Y \given X, T = 1},\quad
\PP^{\rm sup}_{Y(1) \given X,T = 0} = Q_{Y \given X}.
\$
So far the joint distribution of $(X,T,Y(1))$ has 
been determined. We let $U = Y(1)$ be the unobserved confounder.
Finally, the distribution of $Y(0) \given X,T,Y(1),U$ is 
specified via
\$
\PP^{\rm sup}_{Y(0) \given X,T,Y(1),U} = 
\PP^{\rm sup}_{Y(0) \given X} = 
\PP^{\rm obs}_{Y \given X,T=0}.
\$
Having constructed $\PP^{\rm sup}$, we proceed the check
that it satisfies the conditions. Conditional on $X$ and
$U$, $Y(1)$ becomes deterministic and the distribution of 
$Y(0)$ only depends on $X$. Hence $(Y(1),Y(0)) \indep T 
\given X,U$. By construction, it
is straightforward to see that $\PP^{\rm sup}$ is compatible
with $\PP^{\rm obs}_{X,T,Y}$. For any $x$, again by the 
construction of $\PP^{\rm sup}$,
\$
\E_{\PP^{\rm sup}_{Y(1) \given X = x, T=1}}
\bigg[f\Big(\frac{\ud \PP^{\rm sup}_{Y(1) \given X = x,T = 0}}
{\ud \PP^{\rm \sup}_{Y(1) \given X = x, T=1}}\Big)\bigg]
= 
\E_{\PP^{\rm obs}_{Y \given X = x, T=1}}
\bigg[f\Big(\frac{\ud Q_{Y\given X = x}}
{\ud \PP^{\rm \obs}_{Y \given X = x, T=1}}\Big)\bigg]
= D_f\big(Q_{Y \given X = x} ~\|~ \PP^{\rm obs}_{Y \given X =x, T=1}\big)
\le \rho,
\$
where the last inequality is due to the 
definition of $\cQ$.
On the other hand,
\$
\frac{\ud \PP^{\rm sup}_{Y(1) \given X = x,T = 0}}
{\ud \PP^{\rm \sup}_{Y(1) \given X = x, T=1}}
= \frac{\ud \PP^{\rm sup}_{Y(1), T = 0\given X = x}}
{\ud \PP^{\rm \sup}_{Y(1), T =1 \given X = x}}
\cdot 
\frac{\PP^{\rm sup}(T = 1 \given X = x)}
{\PP^{\rm sup}(T = 0 \given X = x)}
= & \frac{\PP^{\sup}(T = 0\given Y(1),X = x)}
{\PP^{\rm sup}(T =1 \given Y(1), X = x)}
\cdot 
\frac{\PP^{\rm sup}(T = 1 \given X = x)}
{\PP^{\rm sup}(T = 0 \given X = x)}\\
= &\frac{\PP^{\rm sup}(T = 0\given U, X = x)}
{\PP^{\rm \sup}(T =1 \given U,X = x)}
\cdot 
\frac{\PP^{\rm sup}(T = 1 \given X = x)}
{\PP^{\rm sup}(T = 0 \given X = x)},
\$
where the last equality is because $U = Y(1)$ 
under $\PP^{\rm sup}$.
Combing the above, we have
\$
\E_{\PP^{\rm sup}_{U \given X = x, T=1}}
\bigg[f\Big(\frac{\PP^{\rm sup}(T = 0\given U, X = x)}
{\PP^{\rm \sup}(T =1 \given U,X = x)}
\cdot 
\frac{\PP^{\rm sup}(T = 1 \given X = x)}
{\PP^{\rm sup}(T = 0 \given X = x)}
\Big)\bigg] \le \rho.
\$
Similarly,
\$
  & \E_{\PP^{\rm sup}_{U \given X = x, T=0}}
\bigg[f\Big(\frac{\PP^{\rm sup}(T = 1\given U, X = x)}
{\PP^{\rm \sup}(T =0 \given U,X = x)}
\cdot 
\frac{\PP^{\rm sup}(T = 0 \given X = x)}
{\PP^{\rm sup}(T = 1 \given X = x)}
\Big)\bigg] \le 
\E_{\PP^{\rm sup}_{Y(1) \given X = x, T=0}}
\bigg[f\Big(\frac{\ud \PP^{\rm sup}_{Y(1) \given X = x,T = 1}}
{\ud \PP^{\rm \sup}_{Y(1) \given X = x, T=0}}\Big)\bigg]\\
= &
\E_{\PP^{\rm obs}_{Y \given X = x, T=0}}
\bigg[f\Big(\frac{\ud \PP^{\obs}_{Y \given X = x, T=1}}
{\ud Q_{Y \given X = x}}\Big)\bigg]
= D_f\big(\PP^{\obs}_{Y \given X = x,T=1} ~\|~ Q_{Y \given X =x}\big)
\le \rho.
\$
Therefore, the super population $\PP^{\rm sup}$
satisfies the $(f,\rho)$-selection condition.
By construction, $\PP^{\rm sup}_{Y(1) \given X, T = 0}
= Q_{ Y\given X}$.
It remains to show that 
$\PP^{\rm sup}_{X\given T = 0} = Q_X$.
For any measurable set $A$, 
\$
\PP^{\rm sup}(X \in A \given T=0) 
= &\E^{\rm sup}\bigg[\frac{\ud \PP^{\rm sup}_{X \given T=0}}
{\ud \PP^{\rm sup}_{X \given T=1}}
\cdot \ind\{X\in A\}\Biggiven T=1\bigg]
= \E^{\rm sup}\bigg[\frac{\ud \PP^{\rm obs}_{X \given T=0}}
{\ud \PP^{\rm obs}_{X \given T=1}}
\cdot \ind\{X\in A\}\Biggiven T=1\bigg]\\
=&\E^{\rm sup}\Big[r_{1,0}(X)
\cdot \ind\{X\in A\}\Biggiven T=1\Big]
= \E^{\rm sup}\bigg[\frac{\ud Q_X}
{\ud \PP^{\rm sup}_{X \given T=1}}
\cdot \ind\{X\in A\}\Biggiven T=1\bigg]
= Q(X \in A).
\$
Since the above holds for any measuable set $A$,
$\PP_{X\given T = 0} = Q_X$.

Finally, switching the role of $1$ and $0$ 
completes the proof.

\end{proof}

\section{Deferred details and discussions}  


\subsection{Proof of Proposition~\ref{prop:positive_alpha}}
\label{app:proof_positive_alpha}
Given $X = x$, suppose instead $\alpha^*(x) = 0$.
We consider the following two cases:
\begin{itemize}
  \item If $\eta^*(x) < -\uly(x)$, then
\$
& \lim\!\inf_{\alpha \rightarrow 0}
\EE\bigg[\alpha f^*\Big(\frac{Y(1) + \eta^*(x)}{-\alpha}\Big)
+ \eta^*(x) + \alpha\rho \biggiven X = x, T=1\bigg]\\
= &
\lim\!\inf_{\alpha \rightarrow 0}
\EE\bigg[\alpha f^*\Big(\frac{Y(1) + \eta^*(x)}{-\alpha}\Big)
\ind \big\{Y(1) \le -\eta^*(x)\big\} \\
&\qquad +\alpha f^*\Big(\frac{Y(1) + \eta^*(x)}{-\alpha}\Big)
\ind \big\{Y(1) > -\eta^*(x)\big\} \biggiven X = x, T=1\bigg]
 + \eta^*(x) + \alpha \rho\\
\stackrel{\rm (a)}{\ge} & 
\lim\!\inf_{\alpha \rightarrow 0}
\EE\bigg[\alpha f^*\Big(\frac{Y(1) + \eta^*(x)}{-\alpha}\Big)
\ind \big\{Y(1) \le -\eta^*(x)\big\} - \alpha L\Biggiven X = x, T=1 \bigg]\\
&\qquad + \lim\!\inf_{\alpha \rightarrow 0}\bigg[
\alpha f^*\Big(\frac{Y(1) + \eta}{-\alpha}\Big)
\ind \big\{Y(1) > -\eta^*(x)\big\} - \alpha L\biggiven X = x, T=1\bigg]
+  \eta^*(x) \\
\stackrel{\rm (b)}{\ge} & + \infty,
\$
where step (a) uses the fact that $\lim\!\inf_{n\rightarrow \infty} a_n + b_n
\ge \lim\!\inf_{n\rightarrow \infty} a_n + \lim\!\inf_{n\rightarrow \infty}
b_n$ and step (b) follows from Fatou's lemma and the condition that 
$f^*(x)/x \rightarrow \infty$ when $x \rightarrow \infty$.

\item If $\eta^*(x) \ge - \uly(x)$, then
\$
&\lim_{\alpha \rightarrow 0}
\EE\bigg[\alpha f^*\Big(\frac{Y(1) + \eta^*(x)}{-\alpha}\Big)
+ \eta^*(x) + \alpha\rho \biggiven X = x, T=1\bigg]\\
= &\lim_{\alpha \rightarrow 0}
\EE\bigg[\alpha f^*\Big(\frac{Y(1) + \eta^*(x)}{-\alpha}\Big)
\ind\big\{Y(1) \ge \ess\!\inf Y(1)\big\}
+ \alpha\rho \biggiven X = x, T=1\bigg] + \eta^*(x)\\
\stackrel{\rm (a)}{=} &\EE\bigg[\lim_{\alpha\rightarrow 0} \alpha f^*\Big(\frac{Y(1) + \eta^*(x)}{-\alpha}\Big)
\ind\big\{Y(1) \ge \ess\!\inf Y(1)\big\}
+ \alpha\rho \biggiven X = x, T=1\bigg] + \eta^*(x)\\
\stackrel{\rm (b)}{=} & \eta^*(x).
\$
Above, step (a) is due to the fact that $f^*(x)$ is 
bounded when $x \le 0$ and the dominated convergence theorem;
step (b) is because $f^*(x)/x \rightarrow 0$ as $x\rightarrow -\infty$.
\end{itemize}
Combining the two cases above, we conclude that 
$\eta^*(x) = -\uly(x)$ and the optimal 
value of the dual problem is $-\uly(x)$.
By the strong duality, the optimal value of the primal 
objective function is $\EE\big[r_{1,0}(X)\uly(X)\given T=1\big]$.
As an implication, there exists a feasible $L(x,y)$ such that 
$\EE[Y(1)L(X,Y(1)) \given  T = 1] = \EE\big[r_{1,0}(X)\uly(X) \given T=1\big]$.
Let $\QQ_{Y \given X = x}$ denote the measure induced by $L(x,y)$:
\$
\frac{\ud \QQ_{Y \given X = x}}{\ud \PP_{Y(1)\given X = x, T=1}}(y)
=\frac{L(x,y)}{r_{1,0}(x)}.
\$
This is a valid transformation of measure because
$L(x,y)$ is feasible.
Then $Y(1) = \uly(X)$ a.s. under $\QQ_{Y \given X}$.
Consequently, 
\$
&1 = \QQ\big(Y(1) = \uly(X) \given X = x\big)
= \EE\bigg[\frac{L\big(x, \uly(x)\big)}{r_{1,0}(x)}
\cdot \ind \big\{Y(1) =  \uly(x)\big\} \Biggiven X = x, T=1\bigg] 
= \frac{L(x,\uly(x))}{r_{1,0}(x)}\cdot \bar{p}(x),\\
& 0 = L(x,Y(1)) \cdot \ind
\{Y(1) > \uly(x)\},~\text{a.s. under }\PP_{Y(1)\given X = x, T=1}.
\$
Again since $L$ is feasible, 
\$
\rho \ge \EE\bigg[f\Big(\frac{L(x,Y(1))}{r_{1,0}(x)}\Big) \Bigggiven
X  = x, T=1\bigg]
= \bar{p}(x) \cdot  f\Big(\frac{1}{\bar{p}(x)}\Big)
+ \big(1-\bar{p}(x)\big) \cdot f(0).
\$
This is a contradiction to the condition.
Hence $\alpha^*(x) >0$.

\subsection{Discussions on Assumption~\ref{assump:sieve} on sieve estimation} 
\label{app:subsec_discuss_sieve}

We provide additional discussion on Assumption~\ref{assump:sieve} 
for sieve estimators in the context of $(X,Y(1))\given T=1$. 
In particular, we 
first justify 
the smoothness of the optimizers when the conditional distributions 
are sufficiently smooth. We then 
verify the technical 
conditions for two choices of $f$-divergences: KL-divergence 
and $\chi^2$-divergence. 
Then we discuss some considerations of relaxing the conditions 
with implementations in practice. 

\paragraph{Smoothness of the optimizers.}
We first provide some justifications for assuming the optimizers 
are continuously differentiable. 
By the strong convexity of $f$, its conjugate $f^*$ is continuous, 
hence without loss of generality 
we always assume the differentiation and expectation are 
exchangeable. 
We also assume the conjugate $f^*$ is sufficiently smooth, 
which is the case for many popular choices of $f$-divergence. 
As we discussed in Proposition~\ref{prop:positive_alpha}, 
under mild conditions, the optimizers $(\alpha^*(x),\eta^*(x))$ 
lies in the interior of $\RR^+\times \RR$. 
The optimizers are thus the solutions to
\$
\big( \alpha^*(x),\eta^*(x) \big) = \textrm{argzero} ~ \nabla_{a,b}\bigg\{ a \EE\Big[ f^* \big({ \textstyle \frac{Y(1)+b}{-a}} \big)  \Biggiven X=x,T=1\Big] + b + a\rho \bigg\},
\$
where the right-hand side takes the form 
\$
F(a,b,x) := \EE\big[  g(Y(1),a,b) \biggiven X=x,T=1\big] \in \RR^2
\$
for some differentiable or smooth function $g$ decided by $f^*$ 
and its derivative $(f^*)'$. 
Thus $F(a,b,x)$ is smooth in $(a,b)$ when $f^*$ is sufficiently smooth. 
Now let us assume the conditional distribution $\PP_{Y(1)\given X=x,T=1}$ 
is smooth; for example, for some $h\in \cX$, 
$\PP_{Y(1)\given X=x+th,T=1} = \PP_{Y(1)\given X=x ,T=1} + t\cdot \PP_h$ 
for some measure $\PP_h$ on $\cY$; 
and similar for higher-order expansions. 
This is a reasonable assumption if we are willing 
to assume that the conditional distributions of $Y(1)$ 
are close for similar covariates. 
Concretely, such condition holds when $Y(1)\given X=x,T=1$ is 
a normal distribution with homoskedastic noise 
and a smooth mean function, or heteroskedastic noise 
with a smooth mean function and smooth standard deviation function, etc. 
When the conditional distributions are 
smooth in $x$, the function $F(a,b,x)$ is also smooth in $x$  
by the linearity of conditional expectation. 
Finally, 
if the derivatives with respect to $a,b$ is always invertible (which is 
the case under mild conditions for the examples we discuss shortly) 
and smooth, 
invoking the Implicit Function Theorem~\citep{rudin1976principles}, 
the minimizer can be smooth in $x$.

\paragraph{KL-divergence.} 
A popular choice for the function $f$ is 
$f(x) = x\log x$, which leads to the KL-divergence~\citep{kullback1951information}. 
The dual function in this case is $f^*(y) = e^{y-1}$, 
and the loss function becomes 
\$
\ell(\theta,x,y) = \alpha(x) e^{\frac{y+\eta(x)}{-\alpha(x)}-1} + \eta(x) + \alpha(x) \rho.
\$
The conditional expectation is 
\$
\EE\big[\ell((a,b),x,Y(1))\biggiven X=x,T=1\big] = 
a \EE\big[e^{\frac{Y(1)-b}{-a}-1} \biggiven X=x,T=1\big] + b + a \rho.
\$
We first look at the strong convexity assumption. 
The conditional expectation is twice differentiable, with 
\$
&\nabla_{a}^2 \EE\big[\ell((a,b),x,Y(1))\biggiven X=x,T=1\big] = \frac{1}{a^3}\EE\Big[\big(Y(1)+b\big)^2 e^{\frac{Y(1)+b}{-a}-1} \Biggiven X=x,T=1\Big],\\
&\nabla_{b}^2 \EE\big[\ell((a,b),x,Y(1))\biggiven X=x,T=1\big] = \frac{1}{a } \EE\Big[e^{\frac{Y(1)+b}{-a}-1} \Biggiven X=x,T=1\Big], \\
&\nabla_{a,b}^2 \EE\big[\ell((a,b),x,Y(1))\biggiven X=x,T=1\big] = -\frac{1}{a^2} \EE\Big[ \big(  Y(1)+b \big)  e^{\frac{Y(1)+b}{-a}-1} \Biggiven X=x,T=1\Big].
\$
Therefore, 
a simple calculation shows that 
as long as 
$Y(1)$ is not deterministic
at $(\alpha^*(x),\beta^*(x))$, the 
Hessian matrix is non-singular. 
%
Also, if the underlying distribution $\PP_{Y(1)\given X=x,T=1}$ 
is continuous in $x$, 
the above derivatives, hence the 
eigenvalues of the Hessian matrix is continuous; 
since $\cX$ is compact, there exists a positive uniform lower bound  
for the smallest eigenvalue of the Hessian matrix, leading to 
strong convexity. 

We then consider the 
continuity condition $|\ell(\theta,x,y) - \ell(\theta^* ,x,y)|\leq \bar\ell(x,y)\|\theta(x)-\theta^*(x)\|_2$ for $\|\theta(x) -\theta^*(x)\|_2<\epsilon$ 
for some sufficiently small $\epsilon>0$, where  
$\|\theta(x) -\theta^*(x)\|_2$ is the Euclidean norm, and 
$\sup_{x\in \cX}\EE[\bar\ell(x,Y(1))^2\given X=x,T=1]< M$ for 
some constant $M>0$. By Taylor expansion, we have 
\$
\ell(\theta,x,y) - \ell(\theta^* ,x,y) = \nabla_\theta \ell( \tilde\theta,x,y) \big(\theta^*(x) - \theta(x)\big) ,
\$
where $\tilde\theta(x)$ lies between $\theta(x)$ and $\theta^*(x)$. 
We note that $\nabla_\theta$ is also a smooth function of $\theta$,
and the gradient is uniform bounded for $\theta(x)$ 
within a neighborhood of $\theta^*(x)$ in terms of Euclidean $L_2$-norm. 
In particular, 
\$
\frac{\partial}{\partial a} \ell((a,b),x,y) = \big(1-{ \textstyle \frac{y+b}{a}} \big) e^{ \frac{y+b}{-a}-1}+ \rho,\quad \frac{\partial}{\partial b} \ell((a,b),x,y) =  1-  e^{ \frac{y+b}{-a}-1}.
\$
For any $\|(a,b) - \theta^*(x)\|_2 \leq \epsilon$ for 
sufficiently small $\epsilon$, we can take 
$\bar\ell(x,y)$ as the uniform upper bound of the 
Euclidean norm 
of the gradient, which has finite second moment 
if $Y(1)$ is not too heavy-tailed. 

Finally, the last condition is that there 
exists a constant $C_1$ such that 
$\EE[\ell(\theta,X,Y(1)) - \ell(\theta^*,X,Y(1))\given T=1] \leq C_1 \|\theta - \theta^*\|_{L_2(\PP_{\cdot\given T=1})}^2$ 
when $\theta\in \Lambda_c^p(\cX)\times\Lambda_c^p(\cX)$ and 
$\|\theta - \theta^*\|_{L_2(\PP_{\cdot\given T=1})}$ is sufficiently small. 
Similar to arguments in the proof of Theorem~\ref{thm:sieve}, 
sufficiently small $\|\theta-\theta^*\|_{L_2(\PP_{\cdot\given T=1})}$ 
implies sufficiently small  $\|\theta-\theta^*\|_{\infty}$ 
for this function class. Therefore, we can consider $\theta\in  \Lambda_c^p(\cX)\times\Lambda_c^p(\cX)$ such that  $\|\theta-\theta^*\|_{\infty}$ 
is sufficiently small.  
With a taylor expansion of the conditional expectation 
of the risk at $(\alpha^*(x),\eta^*(x))$, we have 
\$
&\EE\big[\ell(\theta,x,Y(1))\biggiven X=x,T=1\big] - 
\EE\big[\ell(\theta^*,x,Y(1))\biggiven X=x,T=1\big] \\
&= 
1/2\cdot \nabla_{\theta}^2 \EE\big[\ell(\tilde\theta,x,Y(1))\biggiven X=x,T=1\big] [\theta(x) - \theta*(x),\theta(x)-\theta^*(x)]
\$
since the gradient is zero, where $\tilde\theta(x)$ 
lies between $\theta(x)$ and $\theta^*(x)$. 
Previous derivations have shown that the Hessian 
is continuous; also, 
by the compactness of $\cX$ and continuity of $\theta^*(x)$, 
there is a uniform lower bound $\epsilon>0$ 
for $\alpha^*(x)$. 
Thus, when $\|\theta - \theta^*\|_\infty$ 
is sufficiently small, the Hessian is also bounded. 
Again by the compactness of $\cX$, this bound can be taken 
to be uniform for $x\in \cX$, which leads to the desired condition.

\paragraph{$\chi^2$-divergence.} 
Another popular choice is $f(x)=(x-1)^2$, so that 
$f^*(y) = \frac{1}{4}((y+2)_+^2-1)$. 
The conjugate function is a quadratic function on $[-2,\infty)$
and zero on $(-\infty, -2]$, with continuous 
gradient $(f^*)'(y) = (\frac{s}{2}+1)_+$, 
and 
second-order derivative $(f^*)''(y) = \frac{1}{2} \ind\{y>-2\}$; 
the latter is almost-everywhere (under Lebesgue measure) except $y=2$. 
We now proceed to verify the conditions. The loss function is 
\$
\ell(\theta,x,y) = \frac{\alpha(x)}{4} \Big[\Big(  \frac{y+\eta(x)}{-\alpha(x)}+1\Big)_+^2  -1 \Big] + \eta(x) + \alpha(x) \rho.
\$
Assuming $Y(1)$ does not have 
point measure, the differentiation and expectation are exchangeable, 
and 
\$
&\nabla_{a}^2 \EE\big[\ell((a,b),x,Y(1))\biggiven X=x,T=1\big] = \frac{1}{2a^3}\EE\Big[ \big(Y(1)+b\big)^2\ind\{\textstyle{\frac{Y(1)+b}{-a}}>-2\} \Biggiven X=x,T=1\Big],\\
&\nabla_{b}^2 \EE\big[\ell((a,b),x,Y(1))\biggiven X=x,T=1\big] = \frac{1}{2a } \EE\Big[  \ind\{\textstyle{\frac{Y(1)+b}{-a}}>-2\}  \Biggiven X=x,T=1\Big], \\
&\nabla_{a,b}^2 \EE\big[\ell((a,b),x,Y(1))\biggiven X=x,T=1\big] = -\frac{1}{2a^2} \EE\Big[ \big(Y(1)+b\big)  \ind\{\textstyle{\frac{Y(1)+b}{-a}}>-2\} \Biggiven X=x,T=1\Big].
\$
Also, the gradient is given by 
\$
&\nabla_{a}  \EE\big[\ell((a,b),x,Y(1))\biggiven X=x,T=1\big] =  \EE\Big[\big({\textstyle \frac{Y(1)+b}{-2a}+1} \big)_+^2-1 + {\textstyle \frac{Y(1)+b}{a}}\big( { \textstyle \frac{Y(1)+b}{-2a}+1}\big)_+ \Biggiven X=x,T=1\Big]+\rho,\\
&\nabla_{b}  \EE\big[\ell((a,b),x,Y(1))\biggiven X=x,T=1\big] =-  \EE\Big[  \big( { \textstyle \frac{Y(1)+b}{-2a}+1}\big)_+     \Biggiven X=x,T=1\Big] +1, 
\$
which are both zero at $(a,b) = (\alpha^*(x),\eta^*(x))$. 
The form of the loss function implies that $\alpha^*(x)>0$
for almost all $x$; hence there is a uniform lower bound 
$\epsilon>0$ by the compactness of $\cX$. 
By Cauchy-Schwarz inequality, the Hessian 
at $(a,b) = (\alpha^*(x),\eta^*(x))$ is positive 
$\PP(\frac{Y(1)+\eta^*(x)}{-\alpha^*(x)}>-2\given X=x,T=1)=0$ 
or $(Y(1)+\eta^*(x) - c(x))\ind\{\textstyle{\frac{Y(1)+\eta^*(x)}{-a}}>-2\}=0$ 
almost surely for some $c(x)\in \RR$.  
By the optimality condition, 
the former is 
impossible, and 
the latter is also impossible if $Y(1)$ is not deterministic 
conditional on $X=x$. 
Thus, as long as $Y(1)\given X=x$ is not deterministic 
for almost all $x$, the Hessian is positive definite for all $x\in \cX$. 
By compactness of $\cX$ and the continuity, 
we know that the minimial eigenvalue of the Hessian 
is uniformly lower bounded away from zero, hence 
the strong convexity follows. 

The other two conditions are easy to verify in this case: 
the conjugate function $f^*$ is a truncation of a quadratic function. 
Since truncation is a contraction map, 
these results hold easily by the uniform boundedness 
of second-order derivatives. 
We've thus verified the conditions in Assumption~\ref{assump:sieve}
for $\chi^2$-divergence.

\paragraph{Practical conderations.} 
In practice, we might search for $(\alpha^*(x),\eta^*(x))$ 
within the 
function classes with a bounded range of coefficients  
in the two examples we give, 
leading to a compact function space. 
This is typically assumed in the contexts of $M$-estimators 
and sieve estimators~\citep{van2000asymptotic,geer2000empirical,chen1998sieve,chen2007large}. 
In this case, the regularity conditions are easier to verify 
given the uniform boundedness. 
The function space still provides 
finer and finer approximation to the targets 
if the bounded range enlarges properly with $n$.

\subsection{Estimators for bounds on counterfactual means}
\label{app:subsec_all_cond_bounds}

In this section, we summarize the application 
of the procedure in Section~\ref{subsec:cond_method} 
to estimate other lower and upper bounds on 
counterfactual means. 

\begin{enumerate}[(a)]
\item \textbf{Upper bound of $\EE[Y(1)\given T=0]$:} Let $- \hat\mu_{1,0}^+$ 
be the estimator obtained from the procedure in Section~\ref{subsec:cond_method} with $-Y(1)$ replacing $Y(1)$. Then $\sqrt{n}(\hat{\mu}_{1,0}^+ - \mu_{1,0}^+)\rightsquigarrow N(0,\Var(\phi_{1,+}(X,Y,T)))$, 
with influence function 
\$
\phi_{1,+}(X_i,Y_i,T_i) = \frac{T_i}{p_1} r_{1,0} (X_i )\big[ {H}_{1,+}(X_i,-Y_i(1)) -  {h}_{1,+}(X_i )   \big] + \frac{1-T_i}{p_0}h_{1,+}(X_i),
\$
where $H_{1,+}(x,y) = \alpha_{1,+}^*(x)  f^*\big( \frac{y + \eta_{1,+}^*(x) }{- \alpha_{1,+}^*(x)} \big) + \eta_{1,+}^*(x)  + \alpha_{1,+}^*(x)   \rho$ with $(\alpha_{1,+}^*(x),\eta_{1,+}^*(x))$ 
being the minimizer of $\EE[\alpha   f^*\big( \frac{-Y(1) + \eta  }{- \alpha } \big) + \eta   + \alpha  \rho\given X=x,T=1]$, 
and  $h_{1,+}(x) = \EE[H_{1,+}(X,-Y(1))\given X=x,T=1]$. 
\item \textbf{Lower bound of $\EE[Y(0)\given T=1]$:} 
Let $ \hat\mu_{0,1}^-$ 
be the estimator obtained from the procedure in Section~\ref{subsec:cond_method} 
switching the role of treated and control groups. 
Then $\sqrt{n}(\hat{\mu}_{0,1}^- - \mu_{0,1}^-)\rightsquigarrow N(0,\Var(\phi_{0,-}(X,Y,T)))$
with influence function 
\$
\phi_{0,-}(X_i,Y_i,T_i) = \frac{1-T_i}{p_0}  r_{0,1} (X_i) \big[   {H}_{0,-}(X,Y(0)) -  {h}_{0,-}(X_i)\big] + \frac{T_i}{p_1}h_{0,-}(X_i).
\$
Here $H_{0,-}(x,y) = \alpha_{0,-}^*(x)  f^*\big( \frac{y + \eta_{0,-}^*(x) }{- \alpha_{0,-}^*(x)} \big) + \eta_{0,-}^*(x)  + \alpha_{0,-}^*(x)   \rho$, 
and  $(\alpha_{0,-}^*(x),\eta_{0,-}^*(x))$ 
is  the minimizer of $\EE[\alpha   f^*\big( \frac{ Y(0) + \eta  }{- \alpha } \big) + \eta   + \alpha  \rho\given X=x,T=0]$, 
and  $h_{0,-}(x) = \EE[H_{0,-}(X, Y(0))\given X=x,T=0]$.

\item \textbf{Upper bound of $\EE[Y(0)\given T=1]$:} 
Let $- \hat\mu_{0,1}^+$ 
be the estimator obtained from the procedure in Section~\ref{subsec:cond_method} 
switching the role of treated and control groups and replacing $Y(0)$ with $-Y(0)$. 
Then $\sqrt{n}(\hat{\mu}_{0,1}^+ - \mu_{0,1}^+)\rightsquigarrow N(0,\Var(\phi_{0,+}(X,Y,T)))$ 
with influence function 
\$
\phi_{0,+}(X_i,Y_i,T_i) = \frac{1-T_i}{p_0}  r_{0,1} (X_i ) \big[  {H}_{0,+}(X_i,-Y_i(0)) -  {h}_{0,+}(X_i ) \big]  + \frac{T_i}{p_1}  h_{0,+}(X_i),
\$  
where $H_{0,+}(x,y) = \alpha_{0,+}^*(x)  f^*\big( \frac{y + \eta_{0,+}^*(x) }{- \alpha_{0,+}^*(x)} \big) + \eta_{0,+}^*(x)  + \alpha_{0,+}^*(x)   \rho$ with $(\alpha_{0,+}^*(x),\eta_{0,+}^*(x))$ 
being the minimizer of $\EE[\alpha   f^*\big( \frac{-Y(0) + \eta  }{- \alpha } \big) + \eta   + \alpha  \rho\given X=x,T=0]$, 
and  $h_{0,+}(x) = \EE[H_{0,+}(X,-Y(0))\given X=x,T=0]$. 
\end{enumerate}

\section{Technical proofs}

\subsection{Proof of Proposition~\ref{prop:dual_cond}}
\label{app:subsec_dual_cond}

\begin{proof}[Proof of Proposition~\ref{prop:dual_cond}]
We first claim that solving~\eqref{eq:opt_cond} amounts to solving the following problem 
for each $x$:
\#\label{eq:opt_cond_x}
\min_{L(x) \textnormal{~measurable}}~
                &\E[Y(1) L(x)  \given X=x,T=1]\\
    \textnormal{s.t.}~&\E[L(x) \given X=x, T=1]=r_{1,0}(x)\\
                    &\E[f(L(x)/r_{1,0}(x))\given X=x,T=1]\le \rho.
\#
To be specific, 
denoting  the optimal objective of~\eqref{eq:opt_cond_x} as $\mu(x)$ 
and that of~\eqref{eq:opt_cond} as $\mu_{1,0}^-$, we 
are to show that $\mu_{1,0}^- = \EE[\mu(X)\given T=1]$. 
To see why it is the case, 
suppose $L^*$ is the optimizer of~\eqref{eq:opt_cond}, 
then it is measurable with respect to $X$ and $Y(1)$ 
and satisfies the constraints of~\eqref{eq:opt_cond}.
Then $L(x)(\cdot):= L^*(x,\dot)$ is measurable with respect to $Y(1)$, and satisfy 
the constraints of~\eqref{eq:opt_cond_x}. 
As a result, we have  $\EE[Y(1)L^*(x,Y(1))\given X=x,T=1] \geq \mu(x)$. 
Marginalizing over $X$ yields 
$\mu_{1,0}^- = \EE[Y(1)L(X,Y(1))\given T=1] \geq \EE[\mu(X)\given T=1]$.  
On the other hand, suppose $L^*(x)(\cdot)$ is measurable with respect to $Y(1)$ 
and is the minimizer for~\eqref{eq:opt_cond_x} 
for $\PP_{X\given T=1}$-almost all $x$. 
We let $L(x,y) = L^*(x)(y)$, so that it is measurable with respect 
to $(X,Y(1))$ and satisfy the constraints of~\eqref{eq:opt_cond}. 
Thus we have $\EE[L(X,Y(1))Y(1)\given T=1] = \EE[\mu(X)\given T=1] \geq \mu_{1,0}^-$. 
Combining the two directions leads to the equivalence.

In the following, we solve~\eqref{eq:opt_cond_x}  
and write $\EE_x$ in place of $\EE[\cdot\given X=x,T=1]$ for simplicity. 
Invoking~\citet[Theorem 8.6.1]{luenberger1997optimization} to this convex problem, we have 
\$
\min_{\substack{ \E_x[L]=r_{1,0}(x),\\ \E_x[f(L/r_{1,0}(X))]-\rho \leq 0}} \E_x[Y(1) L(x)] = \max_{\alpha\geq 0,\eta\in \R} ~\varphi(\alpha,\eta,x),
\$
where the Slater's condition is satisfied and strong duality holds, and 
\$
&\varphi(\alpha,\eta,x) = \inf_{L\geq 0~\textrm{measurable}} \calL(\alpha,\eta,L,x), \\ 
&\calL(\alpha,\eta,L,x) = \E_x[Y(1)L(x)] + \eta \E_x[L-r_{1,0}(x)] + \alpha\big( \E_x[f(L/r_{1,0}(x))] -\rho\big).
\$
The minimum of $\cL(\alpha,\eta,L,x)$ is thus given by 
\$
\varphi(\alpha,\eta,x) &= 
\EE_x\bigg[  \min_{z\geq 0} \big\{ Y(1) z + \eta z- \eta r_{1,0}(x) + \alpha f\big(z/r_{1,0}(x)\big) - \alpha\rho  \big\}  \bigg] \\
&= \EE_x\bigg[ -\alpha f^*\Big( \frac{r_{1,0}(x) }{- \alpha}\big(Y(1) + \eta\big) \Big) - \eta r_{1,0}(x) - \alpha \rho  \bigg].
\$
Now we write $\alpha(x)$ and $\eta(x)$ to emphasize its dependency on $x$. 
Therefore, by the equivalence discussed in the beginning, 
we have 
\$
\mu_{1,0}^- &= \EE\Big[  \max_{\alpha(X)\geq 0, \eta(X) \in \RR} \varphi\big(\alpha(X),\eta(X),X \big)   \Biggiven T=1\Big] \\
&= \EE\Big[ \varphi\big(\alpha^*(X),\eta^*(X),X\big)\Biggiven T=1\Big],
\$
where for $\PP_{X\given T=1}$-almost all $x$, 
\$
\big( \alpha^*(x),\eta^*(x) \big) \in \argmax{\alpha\geq 0, \eta\in \RR} ~\EE\bigg[   -\alpha f^*\Big( \frac{r_{1,0}(x) }{- \alpha}\big(Y(1) + \eta\big) \Big) - \eta r_{1,0}(x) - \alpha \rho \bigggiven X=x,T=1\bigg].
\$
With a change-of-variable from $\alpha(x)$ to $\alpha(x) r_{1,0}(x)$, 
we have 
\$
\big( \alpha^*(x)/r_{1,0}(x),\eta^*(x) \big) \in \argmax{\alpha\geq 0, \eta\in \RR} ~\EE\bigg[   -\alpha r_{1,0}(x) f^*\Big( \frac{Y(1) + \eta }{- \alpha} \Big) - \eta r_{1,0}(x) - \alpha r_{1,0}(x) \rho \bigggiven X=x,T=1\bigg].
\$
The minimum of~\eqref{eq:opt_cond} can thus be written as 
\$
\mu_{1,0}^- = - \EE\bigg[ r_{1,0}(x) \Big\{  \alpha^*(X)  f^*\Big( \frac{Y(1) + \eta^*(X) }{- \alpha^*(X)} \Big) + \eta^*(X)  + \alpha^*(X)   \rho \Big\} \bigggiven  T=1\bigg],
\$
where for $\PP_{X\given T=1}$-almost all $x$, it holds that 
\$
\big( \alpha^*(x),\eta^*(x) \big) \in \argmin{\alpha\geq 0, \eta\in \RR} ~\EE\bigg[  \alpha  f^*\Big( \frac{Y(1) + \eta }{- \alpha} \Big) + \eta + \alpha  \rho \bigggiven X=x,T=1\bigg].
\$
Therefore, we complete the proof of Proposition~\ref{prop:dual_cond}. 
\end{proof}

\subsection{Proof of convergence of sieve estimator}
\label{app:subsec_sieve}

\begin{proof}[Proof of Theorem~\ref{thm:sieve}]
We analyze the behavior of $\hat\theta^{(j)}$ for each fold $j$. 
As $|\cI_1^{(j)}|\asymp n$, we take the generic notation 
of $\hat\theta$ and sample size $n$, so that 
\$
\hat\EE_n \big[\ell(\hat\theta,X,Y(1))\big] \geq \inf_{\theta\in \Theta_n}  ~ \hat\EE_n \big[\ell(\theta,X,Y(1))\big]  - O_P((\frac{\log n}{n})^{2p/(2p+d)}),
\$
where $(X_i,Y_i) \sim \PP_{X,Y(1)\given T=1}$ are i.i.d.~data. 
For some fixed $b>0$, we denote the sequence 
\$
\delta_n := \inf\bigg\{  \delta \in (0,1)\colon \frac{1}{\sqrt{n}\delta^2} \int_{b\delta^2}^{\delta} \sqrt{\log N(\epsilon^{1+d/2p}, \Theta_n, \|\cdot\|_{L_2(\PP_{\cdot \given T=1})}) \ud \epsilon } \leq 1   \bigg\},
\$
where $N(\epsilon, \Theta_n, \|\cdot\|_{L_2(\PP_{\cdot\given T=1})})$ 
is the $\epsilon$-covering number of $\Theta_n$ in the $L_2$-norm 
under $\PP_{\cdot \given T=1}$. 
We employ the established convergence results for sieve estimators 
adapted from \citet[Theorem 3.2]{chen2007large} and \citet[Lemma B.3]{yadlowsky2018bounds}, stated in Lemma~\ref{lem:sieve_rate}. 

\begin{lemma}
    \label{lem:sieve_rate}
Let $\theta^*\in \Theta$ be a population risk minimizer. 
Supose  
there exists constants $c_1,c_2>0$ such that 
$c_1 \EE[\ell(\theta,X,Y)- \ell(\theta^*,X,Y)] \leq d(\theta,\theta^*)^2 
\leq c_2 \EE[\ell(\theta,X,Y)- \ell(\theta^*,X,Y)]$ 
for $\theta$ in a neighborhood of $\theta^*$. 
Suppose the following conditions hold: 
\begin{enumerate}[(i)]
    \item For sufficiently small $\epsilon>0$, $\Var(\ell(\theta,X,Y)- \ell(\theta^*,X,Y))\leq C_1\epsilon^2$ for all $\theta\in \Theta_n$ 
    such that $d(\theta,\theta^*)\leq \epsilon$. 
    \item For any $\delta>0$, there exists a constant $s\in (0,2)$ 
    and a measurable function $U_n(\cdot)$ such that $\sup_n \EE[U_n(X,Y)^2] \leq C_3$ and 
    $\sup_{\theta\in \Theta_n\colon d(\theta,\theta^*)\leq \delta}|\ell(\theta,X,Y) - \ell(\theta^*,X,Y)|\leq \delta^s U_n(X,Y)$  
    for  constant $C_3>0$. 
\end{enumerate}
Then $d(\hat\theta_n - \theta^*) = O_P(\max\{\delta_n, \inf_{\theta'\in \Theta_n}d(\theta',\theta^*))$. 
\end{lemma}

We define the distance 
as $L_2$-norm 
$d(\theta,\theta') = \|\theta-\theta'\|_{L_2(\PP)}$, 
and verify the conditions in Lemma~\ref{lem:sieve_rate}. 
We define $\Theta = \Lambda_c^p(\cX ) \times \Lambda_c^p(\cX)$ 
without truncation.  
The upper bound $\EE[\ell(\theta,X,Y)-\ell(\theta^*,X,Y)\given T=1]$ is 
directly implied by Assumption~\ref{assump:sieve}.  
By the $\lambda$-strong convexity of  $\EE\big[ \ell((a,b),x,Y) \biggiven X=x  \big]$ 
is  at $(a,b) = \theta^*(x)$,  
\$
\EE\big[ \ell(\theta (x),x,Y) \biggiven X=x  \big] - \EE\big[ \ell(\theta^*(x),x,Y) \biggiven X=x  \big] \geq \lambda \big(\theta(x)-\theta^*(x)\big)^2.
\$
Integrating over $X$ yields $\EE [ \ell(\theta (X),X,Y) ] - \big[ \ell(\theta^*(X),X,Y) ] \geq c'' d(\theta,\theta^*)$ for some constant $c''>0$. 

We then check condition (i).  
By the positive density condition, we have $\|\cdot\|_{L_2(\lambda)} \asymp \|\cdot\|_{L_2(\PP)}$. 
Hence $\|\theta - \theta^*\|_{\infty} = o(1)$ 
once $\|\theta-\theta^*\|_{L_2(\PP)} = o(1)$. 
By Lemma 2 of~\citet{chen1998sieve}, we have 
$\|\theta\|_\infty \lesssim \|\theta\|_{L_2(\lambda)}^{2p/(2p+d)}$ for any 
$\theta\in \Theta$, where $\lambda$ is the Lebesgue measure.  
Therefore, 
sufficiently small $\|\theta-\theta^*\|_{L_2(\PP)}$  
implies sufficiently small $\|\theta-\theta^*\|_{\infty}$. 
Since 
for $\|\theta-\theta^*\|_{\infty}$ sufficiently small, 
$|\ell(\theta,x,y)- \ell(\theta^*,x,y)|\leq \bar{\ell}(x,y) (\theta(x)-\theta^*(x))$ where $\EE[\bar\ell(x,Y)^2\given X=x] \leq M$ for all $x$, 
we have 
\$
\Var\big(  \ell(\theta,X,Y)- \ell(\theta^*,X,Y) \big) 
\leq \EE\big[ |\ell(\theta,x,y)- \ell(\theta^*,x,y)|^2        \big] 
\leq M \EE\big[ \big(\theta(X)-\theta^*(X))^2   \big] \leq M \epsilon^2
\$
for all $\theta\in \Theta_n$ such that $d(\theta,\theta^*)\leq \epsilon$ 
for sufficiently small $\epsilon>0$. 
Condition (ii) follows from the same argument by taking 
$U_n(x,y) = \bar\ell(x,y)$. 
Therefore, applying Lemma~\ref{lem:sieve_rate} we have 
$\|\hat\theta_n - \theta^*\|_{L_2(\PP)} = O_P(\max\{\delta_n ,\inf_{\theta'\in \Theta_n}d(\theta',\theta^*)\})$. Here according to~\citet{chen1998sieve} 
and~\citet{geer2000empirical}, we have 
\$
\log N\big( \epsilon, \Theta_n^\eta, \|\cdot \|_{2,\PP}) \lesssim \textrm{dim}(\Theta_n^\eta) \log \frac{1}{\epsilon},
\$
where $\textrm{dim}(\Theta_n^\eta) = J_n^p$. 
Since truncation is a contraction map, the 
covering number of $\Theta_n^\alpha $ is upper bounded by the above quantity. 
As a result, we have 
\$
\log N\big( \epsilon, \Theta_n , \|\cdot \|_{2,\PP}) \lesssim J_n^p \log \frac{1}{\epsilon}.
\$
Similar to the results in~\citet{yadlowsky2018bounds}, we have 
\$
\delta_n \asymp \sqrt{\frac{J_n^d \log n}{n}}.
\$
We finally bound the approximation error using $\Theta_n$. 
Note that we take $\Theta_n$ to be truncated at $\epsilon$. 
However, since the population minimizer $\theta^*$ is uniformly  
bounded above $\epsilon$, since 
truncation is a contraction map, 
we have 
$\inf_{\theta\in \Theta_n} \|\theta-\theta^*\|_{L_2(\PP)} 
\leq \inf_{\theta \in \Theta_n^\eta \times \Theta_n^\eta} \|\theta-\theta^*\|_{L_2(\PP)}  \leq O(J_n^p)$, where 
the last inequality is a well-established result, 
see, e.g.,~\citet{timan2014theory}.  
We now set $J_n = (n/\log n)^{1/(2p+d)}$, so that 
$\|\hat\theta - \theta^* \|_{L_2(\PP)} = O_P((\log n/n)^{p/(2p+d)})$. 
This completes our proof. 
\end{proof}

\subsection{Proof of Theorem~\ref{thm:double_consist}}
\label{app:subsec_double_consist}

\begin{proof}[Proof of Theorem~\ref{thm:double_consist}]
We consider the general scenario where $(\hat\alpha^{(j)},\hat\eta^{(j)})$ 
converges in sup-norm to some fixed $(\alpha^\diamond, \eta^\diamond)$, 
and show that 
$-\hat\mu_{1,0}^{(j)}\stackrel{P}{\to} \EE[r(X)\ell(\theta^\diamond(X),X,Y(1)\given T=1]$ for any fixed $j$, 
where  the risk function $\ell$ is 
defined in Proposition~\ref{prop:dual_cond}. 
In the following, we drop the dependency on $j$ for notational convenience; 
we are to show that with estimators $\hat{r}$, $\hat{H}$ and $\hat{h}$ 
that are independent of $\cI_1$ and $\cI_0$, 
\#\label{eq:consist_target}
\hat\mu:= \frac{1}{|\cI_1 |} \sum_{i\in \cI_1 }
\hat{r} (X_i) \big(\hat{H} (X_i,Y_i) - \hat{h} (X_i)\big) + \frac{1}{|\cI_0 |} \sum_{i\in \cI_0 } \hat{h} (X_i) \stackrel{P}{\to} \EE\big[r(X )   \ell ( \theta^\diamond(X), X ,Y(1))  \biggiven T=1\big].
\#
Therefore, if $\theta^\diamond = \theta^*$, we have $\hat\mu_{1,0} = \mu_{1,0}^-+o_P(1)$
since $\mu_{1,0}^- = - \EE\big[r(X )   \ell ( \theta^*(X), X ,Y(1))  \biggiven T=1\big]$ 
by Proposition~\ref{prop:dual_cond}. 
Otherwise, since 
$\EE\big[r(X )   \ell ( \theta^*(X), X ,Y(1))  \biggiven T=1\big]\leq \EE\big[r(X )   \ell ( \theta^\diamond(X), X ,Y(1))  \biggiven T=1\big]$, 
we have the one-sided validity that 
$\hat\mu_{1,0} \stackrel{P}{\to} - \EE\big[r(X )   \ell ( \theta^\diamond(X), X ,Y(1))  \biggiven T=1\big] \leq \mu_{1,0}^-$, i.e., 
our estimator converges to a valid lower bound. 

It thus remains to show~\eqref{eq:consist_target}. We prove the results 
when either  $\hat{r}$ or $\hat{h}$ is consistent.

\paragraph{Consistent $\hat{r}$.}
We first show the case where $\hat{r}$ is consistent for $r_{1,0}$,
but not necessarily the regression function $\hat{h}$. 
Recall that $\hat{H}(x,y) = \ell(\hat\theta(x),x,y)$. 
Note that 
\$
\hat\mu = \frac{1}{|\cI_1 |} \sum_{i\in \cI_1 }
\big(\hat{r} (X_i) - r(X_i) \big)\big(\hat{H} (X_i,Y_i) - \hat{h} (X_i)\big) + \frac{1}{|\cI_1 |} \sum_{i\in \cI_1 }
  r(X_i)  \big(\hat{H} (X_i,Y_i) - \hat{h} (X_i)\big) + \frac{1}{|\cI_0 |} \sum_{i\in \cI_0 } \hat{h} (X_i).
\$
The first summation can be controlled as (where the expectation is  
implicitly conditional on other folds except $\cI_0^{(j)}\cup \cI_1^{(j)}$)
\$
&\EE\bigg[\Big(\frac{1}{|\cI_1 |} \sum_{i\in \cI_1 }
\big(\hat{r} (X_i) - r_{1,0}(X_i) \big)\big(\hat{H} (X_i,Y_i) - \hat{h} (X_i)\big) \Big)^2  \bigg] \\ 
&\leq \EE\Big[   \big(\hat{r} (X_i) - r_{1,0}(X_i) \big)^2 \EE\big[ \big(\hat{H} (X_i,Y_i) - \hat{h} (X_i)\big)^2 \biggiven X_i\Big]\bigg] \leq M\cdot \|\hat{r} - r_{1,0}\|_{L_2(\PP_{X\given T=1})}^2 = o_P(1).
\$
Invoking Lemma~\ref{lem:cond_to_op}, we can drop the conditioning and the first summation is $o_P(1)$. 
On the other hand, since the covariate shift between $\PP_{X\given T=1}$ and $\PP_{X\given T=0}$ is 
exactly $r_{1,0}$, we know that $\EE[\hat{h}(X)r(X)\given T=1] = \EE[\hat{h}(X)\given T=0]$, 
where we still implicitly condition on other folds. As a result, 
\$
&- \frac{1}{|\cI_1 |} \sum_{i\in \cI_1 }
  r_{1,0}(X_i)   \hat{h} (X_i) + \frac{1}{|\cI_0 |} \sum_{i\in \cI_0 } \hat{h} (X_i)\\
&= - \frac{1}{|\cI_1 |} \sum_{i\in \cI_1 }
  \big( r_{1,0}(X_i)   \hat{h} (X_i) -\EE\big[\hat{h}(X)r(X)\biggiven T=1\big] \big) + \frac{1}{|\cI_0 |} \sum_{i\in \cI_0 }\big(  \hat{h} (X_i) - \EE\big[ r(X)\biggiven T=0\big] \big),
\$
where both terms are unbiased. Thus by Cauchy-Schwarz inequality, 
\$
&\EE\bigg[\Big(- \frac{1}{|\cI_1 |} \sum_{i\in \cI_1 }
  r_{1,0}(X_i)   \hat{h} (X_i) + \frac{1}{|\cI_0 |} \sum_{i\in \cI_0 } \hat{h} (X_i)\Big)^2  \bigg]\\
&\leq \frac{2}{|\cI_1|} \Var\big( r_{1,0}(X )   \hat{h} (X ) \biggiven T=1 \big) + \frac{2}{|\cI_0|} \Var\big( \hat{h} (X ) \biggiven T=0 \big) = o_P(1)
\$
invoking the assumption that $\hat{h}$ as finite second moment. 
Drop the conditioning by Lemma~\ref{lem:cond_to_op}, we know that this summation is also $o_P(1)$, 
hence 
\$
\hat\mu &= \frac{1}{|\cI_1 |} \sum_{i\in \cI_1 }
  r(X_i)  \hat{H} (X_i,Y_i)  + o_P(1) \\
  &= \frac{1}{|\cI_1 |} \sum_{i\in \cI_1 }
  r(X_i)   \ell\big( \theta^\diamond(X_i), X_i,Y_i\big) 
  +\frac{1}{|\cI_1 |} \sum_{i\in \cI_1 }
  r(X_i) \big\{  \ell\big(\hat\theta(X_i), X_i,Y_i\big) - \ell\big( \theta^\diamond(X_i), X_i,Y_i\big) \big\} + o_P(1)\\ 
  &= \EE\big[r(X_i)   \ell ( \theta^\diamond(X), X ,Y(1))  \biggiven T=1\big]+\frac{1}{|\cI_1 |} \sum_{i\in \cI_1 }
  r(X_i) \big\{  \ell\big(\hat\theta(X_i), X_i,Y_i\big) - \ell\big( \theta^\diamond(X_i), X_i,Y_i\big) \big\} + o_P(1).
\$ 
Finally, once $ \|\hat\theta -\theta^\diamond \big\|_{\infty,\PP_{X\given T=1}} = o_P(1)$, 
by the local expansion around $\theta^\diamond(x)$, we have 
\$
\Big|\ell\big(\hat\theta(X_i), X_i,Y_i\big) - \ell\big( \theta^\diamond(X_i), X_i,Y_i\big)\Big| 
\leq M(X_i,Y_i) \big\|\hat\theta(X_i)-\theta^\diamond(X_i)\big\|_2,
\$
hence (implicitly conditioning on other folds) we have 
\$
 &\EE\bigg[ \Big(\frac{1}{|\cI_1 |} \sum_{i\in \cI_1 }
  r(X_i) \big\{  \ell\big(\hat\theta(X_i), X_i,Y_i\big) - \ell\big( \theta^\diamond(X_i), X_i,Y_i\big) \big\}  \Big)^2\bigg] \\ 
  &\leq \EE\Big[ r(X)^2 M(X ,Y(1))^2   \big\|\hat\theta(X_i)-\theta^\diamond(X_i)\big\|_2^2  \Big] = o_P(1)
\$
since $\EE[M(X,Y(1))^2\given T=1]\leq M$ for some constant $M>0$. 
We've thus completed the proof of~\eqref{eq:consist_target}.

\paragraph{Consistent $\hat{h}$.} We then show the results when $\hat{h}$ is consistent, 
but not necessarily $\hat{r}$. 
In this case, $\|\hat{h}-\bar{h}\|_{L_2(\PP_{X\given T=1})}=o_P(1)$, 
where $\bar{h} = \EE[\hat{H}(X,Y(1))\given X=x,T=1]$ viewing $\hat{H}$ as fixed. 
Note that 
\$
\hat\mu = \frac{1}{|\cI_1 |} \sum_{i\in \cI_1 }
\hat{r} (X_i) \big(\hat{H} (X_i,Y_i) - \bar{h}(X_i) \big)  + \frac{1}{|\cI_1 |} \sum_{i\in \cI_1 }\hat{r}(X_i) \big( \bar{h}(X_i) -  \hat{h}(X_i) \big)  + \frac{1}{|\cI_0 |} \sum_{i\in \cI_0 } \hat{h} (X_i).
\$
The first summation is unbiased conditional on other folds, hence 
\$
\EE\bigg[ \Big( \frac{1}{|\cI_1 |} \sum_{i\in \cI_1 }
\hat{r} (X_i) \big(\hat{H} (X_i,Y_i) - \bar{h}(X_i) \big)\Big)^2\bigg] 
= \frac{1}{|\cI_1|} \Var\big( \hat{r} (X) \{\hat{H} (X,Y(1)) - \bar{h}(X ) \} \biggiven T=1 \big) = o_P(1)
\$
due to the finite second moments. By Cauchy-Schwarz inequality, the second summation satisfies 
\$
\EE\bigg[\Big(\frac{1}{|\cI_1 |} \sum_{i\in \cI_1 }\big( \bar{h}(X_i) -  \hat{h}(X_i) \big) \Big)^2\bigg]
\leq \big\|\hat{r}\cdot(\hat{h} - \bar{h})\big\|_{L_2(\PP_{X\given T=1})}^2 = o_P(1)
\$
due to the boundedness of $\hat{r}$. 
Similarly, we know that $\frac{1}{|\cI_0 |} \sum_{i\in \cI_0 } \bar{h} (X_i) - \hat{h}(X_i) = o_P(1)$. 
Consequently, invoking Lemma~\ref{lem:cond_to_op} we drop the implicit conditioning 
and arrive at 
\$
\hat\mu &= \frac{1}{|\cI_0 |} \sum_{i\in \cI_0 } \bar{h} (X_i) + o_P(1) = \EE\big[ \bar{h}(X)\given T=0] + o_P(1)
\$  
further using the finite second moment of $\bar{h}$ (or $\hat{H}$), where we implicitly condition on other folds 
and view $\hat{h}$ as fixed. Finally, denoting $h^\diamond (x) = \EE[\ell(\theta^\diamond(x),x,Y(1))\given X=x,T=1]$, we note that by Jensen's inequality, 
\$
\Big|\EE\big[ \bar{h}(X)\given T=0] - \EE\big[ h^\diamond (X)\biggiven T=0 \big]\Big|^2
\leq \|\bar h - h^\diamond \|_{L_2(\PP_{X\given T=0})}^2 
\leq \big\| \hat{H}(X,Y(1)) - \ell(\theta^\diamond(X),X,Y(1))\big\|_{L_2(\PP_{\cdot\given T=0})}^2.
\$
By the same argument as the previous case and due to the uniform boundedness of 
the covariate shift $r_{1,0}(\cdot)$, the above term is $o_P(1)$. 
Therefore, by the change-of-measure with $r_{1,0}$, we have 
\$
\hat\mu &=  \EE\big[  {h}^\diamond(X)\given T=0] + o_P(1)
= \EE\big[  r(X) {h}^\diamond(X)\given T=1] + o_P(1) 
= \EE\big[r(X )   \ell ( \theta^\diamond(X), X ,Y(1))  \biggiven T=1\big] + o_P(1)
\$
by the tower property of conditional expectations. 
We thus complete  the proof of two cases 
and  conclude the proof of Theorem~\ref{thm:double_consist}. 
\end{proof}

\subsection{Proof of Theorem~\ref{thm:cond}}
\label{app:subsec_thm_cond}

\begin{proof}[Proof of Theorem~\ref{thm:cond}]
We  show that for each $j$, we have 
$\hat\mu_{1,0}^{(j)} = \hat\mu_{1,0}^{*,(j)}+o_P(1/\sqrt{n})$, 
where 
\$
\hat\mu_{1,0}^{*(j)} = \frac{1}{|\cI_1^{(j)}|} \sum_{i\in \cI_1^{(j)}}
r_{1,0} (X_i) \big( {H}(X_i,Y_i) -  {h}(X_i)\big) + \frac{1}{|\cI_0^{(j)}|} \sum_{i\in \cI_0^{(j)}} {h}(X_i),
\$
and we define 
\$
H(x,y) = \alpha^*(x)  f^*\Big( \frac{y + \eta^*(x) }{- \alpha^*(x)} \Big) + \eta^*(x)  + \alpha^*(x)   \rho,\quad h(x) = \EE\big[H(X,Y(1))\biggiven X=x,T=1\big].
\$
We show this result for any $j$; 
we implicitly condition on all the remaining folds 
other than $\cI_1^{(j)}$ and $\cI_0^{(j)}$, 
so that all nuisance components are viewed as fixed. 
To simplify notations, we write $\cI_1 := \cI_1^{(j)}$, 
$\cI_0 := \cI_0^{(j)}$ and $r:=r_{1,0}$, $\hat{r}:= \hat{r}^{(j)}$, $\hat{h}:= \hat{h}^{(j)}$, $\hat{H}:= \hat{H}^{(j)}$, $\bar{h}:=\bar{h}^{(j)}$. 
We also represent the parameters (functionals) with 
\$
\hat\theta(\cdot) = \big(\hat\alpha(\cdot),\hat\eta(\cdot)\big) 
:= \big(\hat\alpha^{(j)}(\cdot), \hat\eta^{(j)}(\cdot)\big),\quad 
\theta^*(\cdot) := \big( \alpha^*(\cdot), \eta^*(\cdot)\big),
\$
and recall the generic function (where $\theta = (\alpha(\cdot),\eta(\cdot))$)
\$
\ell(\theta,x,y) = \alpha(x) f^*\Big( \frac{y+\eta(x)}{-\alpha(x)} \Big) + \eta(x) + \alpha(x) \rho,
\$
so that $H(x,y) =\ell(\theta^*,x,y)$ and $\hat{H}(x,y) = \ell(\hat\theta,x,y)$. 
By definition, 
we have the decomposition 
\$
\hat\mu_{1,0}^{(j)} - \hat\mu_{1,0}^{*(j)} 
&= \frac{1}{|\cI_1|} \sum_{i \in \cI_1 } \Big[ \hat{r}(X_i) \big(\hat{H}(X_i,Y_i) - \hat{h}(X_i) \big) - r (X_i) \big( H(X_i,Y_i) - h(X_i) \big) \Big] + 
\frac{1}{|\cI_0|} \sum_{i\in \cI_0} \big( \hat{h}(X_i) - h(X_i)\big)\\
&= \frac{1}{|\cI_1|} \sum_{i \in \cI_1 }  
r(X_i) \big(\hat{H}(X_i,Y_i) -  H(X_i,Y_i)\big) 
- \frac{1}{|\cI_1|} \sum_{i \in \cI_1 }  \big(\hat{r}(X_i)-r(X_i)\big) \big(\hat{h}(X_i) - \bar{h}(X_i)\big) \\
&\qquad + \frac{1}{|\cI_1|} \sum_{i \in \cI_1 }  \big(\hat{r}(X_i)-r(X_i)\big) \big( \hat{H}(X_i,Y_i) - \bar{h}(X_i) \big) \\
&\qquad - \frac{1}{|\cI_1|} \sum_{i \in \cI_1 }   r(X_i) \big(\hat{h}(X_i) - h(X_i)\big)  + 
\frac{1}{|\cI_0|} \sum_{i\in \cI_0} \big( \hat{h}(X_i) - h(X_i)\big).
\$
In the following, we are to bound the 
several summations sparately. 
Firstly, 
by Cauchy-Schwarz inequality, 
\$
\bigg| \frac{1}{|\cI_1|} \sum_{i \in \cI_1 }  \big(\hat{r}(X_i)-r(X_i)\big) \big(\hat{h}(X_i) - \bar{h}(X_i)\big)\bigg|
&\leq \sqrt{\frac{1}{|\cI_1|} \sum_{i \in \cI_1 }  \big(\hat{r}(X_i)-r(X_i)\big)^2 }\sqrt{\frac{1}{|\cI_1|} \sum_{i \in \cI_1 }  \big(\hat{h}(X_i) - \bar{h}(X_i)\big)^2 } \\
&= O_P\big( \|\hat{r}-r\|_{L_2(\PP_{X\given T=1})} \cdot \|\hat{h}-h\|_{L_2(\PP_{X\given T=1})} \big) = o_P(1/\sqrt{n})
\$
under the given convergence rate of the product. 
Since $\bar{h}(x) = \EE\big[\hat{H}(X,Y(1))\given X=x,T=1]$ 
for the fixed function $\hat{H}$, the term 
$\big(\hat{r}(X_i)-r(X_i)\big) \big(\hat{H}(X_i,Y_i) - \bar{h}(X_i)\big)$ 
has mean zero, hence by Markov's inequality, 
\$
\frac{1}{|\cI_1|} \sum_{i \in \cI_1 }  \big(\hat{r}(X_i)-r(X_i)\big) \big(\hat{h}(X_i) - \bar{h}(X_i)\big) = O_P\big(\sqrt{ \Var(\hat{r}(X_i)-r(X_i) )  (\hat{H}(X_i,Y_i)- \bar{h}(X_i) ))}/\sqrt{n}\big),
\$
where by the consistency of $\hat{r}$, this term is $o_P(1/\sqrt{n})$. 
Furthermore, note that
\#\label{eq:h_contrast}
& \frac{1}{|\cI_1|} \sum_{i \in \cI_1 }   r(X_i) \big(\hat{h}(X_i) - h(X_i)\big)  - 
\frac{1}{|\cI_0|} \sum_{i\in \cI_0} \big( \hat{h}(X_i) - h(X_i)\big) \notag\\
&= \frac{1}{|\cI_1|} \sum_{i \in \cI_1 } \Big(   r(X_i) \big(\hat{h}(X_i) - h(X_i)\big)
- \EE\big[ r(X) (\hat{h}(X)-h(X))\biggiven T=1 \big] \Big) \notag \\
& \qquad - 
\frac{1}{|\cI_0|} \sum_{i\in \cI_0} \big( \hat{h}(X_i) - h(X_i) - \EE[\hat{h}(X_i) - h(X_i)\given T=0]\big),
\#
where we use the equivalence of the two expectations: 
this is because there is a covariate shift $r(X)$ from $\PP_{X\given T=1}$ 
to $\PP_{X\given T=0}$, 
hence $\EE[\phi(X)r(X)\given T=1]= \EE[\phi(X)\given X=0]$ 
for any integrable function $\phi\colon \cX\to \RR$. 
The two summations in~\eqref{eq:h_contrast} is thus 
both unbiased, indicating 
\$
& \frac{1}{|\cI_1|} \sum_{i \in \cI_1 }   r(X_i) \big(\hat{h}(X_i) - h(X_i)\big)  - 
\frac{1}{|\cI_0|} \sum_{i\in \cI_0} \big( \hat{h}(X_i) - h(X_i)\big) \notag\\
&= O_P\big(\sqrt{\Var(r(X)(\hat{h}(X)-h(X))\given T=1)/n} 
+ \sqrt{\Var(r(X)(\hat{h}(X)-h(X))\given T=0)/n}\big) \\
&= O_P \big(\|r(X)(\hat{h}(X)-h(X))\|_{L_2(\PP_{X\given T=1})}/\sqrt{n} + \|\hat{h}-h\|_{L_2(\PP_{X\given T=0})}/\sqrt{n}\big) = o_P(1/\sqrt{n}),
\$
where the last equality follows from 
the $L_2$-consistency of $\hat{h}$ to $\bar{h}$ and 
the fact that $\|\bar{h} - h\|_{L_2(\PP_{X\given T=1})}$ by 
the stability of the conditional expectations
induced by the stability of $g$ in Assumption~\ref{assump:cond_regularity}.
Finally, we turn to 
\$
\frac{1}{|\cI_1|} \sum_{i \in \cI_1 }  
r(X_i) \big(\hat{H}(X_i,Y_i) -  H(X_i,Y_i)\big).
\$
Since $\ell(\theta,x,y)$ is a convex function in $\theta$ for any $(x,y)$, 
for any ($\PP_{X\given T=1}$-almost all) $x$, 
\$
\EE\big[ \ell(\theta,x,y) \biggiven X=x,T=1\big] = \alpha(x) \EE\bigg[f^*\Big( \frac{Y(1)+\eta(x)}{-\alpha(x)} \Big)\bigggiven X=x,T=1\bigg] + \eta(x) + \alpha(x) \rho 
\$
is also convex and differentiable by the given regularity condition. 
In particular, by the optimality of $(\alpha^*(x),\eta^*(x))$ 
for the per-$x$ minimization problem and the 
exchangeability of differentiation and expectation, 
\$
\nabla_\theta\EE\big[ \ell(\theta,x,Y(1)) \biggiven X=x,T=1\big] \big|_{\theta = (\alpha^*(x),\eta^*(x))} = \EE\big[ \nabla_\theta \ell(\theta^*(x),x,Y(1)) \biggiven X=x,T=1\big] = 0. 
\$
Multiplying $r(X)$ and integrating over $X\given T=1$, we know that 
\#\label{eq:zero_grad}
\EE\big[ r(X)\nabla_\theta \ell(\theta^*(X),X,Y(1)) [\hat\theta(X) - \theta^*(X) ]  \biggiven T=1\big] = 0.
\#
By Lemma 2 of~\citet{chen1998sieve}, when both $\hat\theta$ and $\theta^*$ is smooth enough,  
sufficiently small $\|\hat\theta-\theta^*\|_{L_2(\PP)}$  
implies sufficiently small $\|\hat\theta-\theta^*\|_{\infty}$. 
As a result, when $\|\hat\theta(x) - \theta^*(x)\|_{L_2(\PP_{\cdot \given T=1})}$ is sufficiently small, 
by the condition that 
$\big|\ell(\hat\theta,x,y) - \ell(\theta^*,x,y) - \nabla_\theta \ell(\theta^*(x),x,y)[\theta^*(x)-\hat\theta(x)]\big|\leq M(x,y) \|\hat\theta(x) - \theta^*(x)\|_2^2$ 
as well as Jensen's inequality, 
we have 
\$
&\Big|  \EE\big[ \hat{H}(X,Y(1)) - H(X,Y(1)) \biggiven T=1\big] \Big| \\
&= \bigg| \EE\Big[  r(X )\Big\{ \ell\big(\hat\theta,X,Y(1)\big) - \ell\big(\theta^*,X,Y(1)\big) - \nabla_\theta \ell\big(\theta^*(X),X,Y(1)\big)\big[\hat\theta(X) - \theta^*(X)\big] \Big\} \Biggiven T=1\Big] \bigg| \\
&\leq  \EE\Big[ r(X)\big| \ell\big(\hat\theta,X,Y(1)\big) -  \ell\big(\theta^*,X,Y(1)\big) - \nabla_\theta \ell\big(\theta^*(X),X,Y(1)\big)\big[\hat\theta(X) - \theta^*(X)\big] \big| \Biggiven T=1\Big]    \\
&\leq    \EE\big[ r(X)M(X,Y(1)) \|\hat\theta(X) - \theta^*(X)\|_{2}^2 \biggiven T=1 \big]
= M \cdot \| r(  \hat\theta - \theta ) \|_{L_2(\PP_{X\given T=0})}^2. 
\$
Returning to our problem, we note that 
due to unbiasedness, 
\$
&\frac{1}{|\cI_1|} \sum_{i \in \cI_1 }  
\Big[ 
r(X_i) \big(\hat{H}(X_i,Y_i) -  H(X_i,Y_i)\big)
-  \EE\big[ \hat{H}(X,Y(1)) - H(X,Y(1)) \biggiven T=1\big] \Big]\\
&= O_P\big(\|r(X)\hat{H}(X,Y(1))-H(X,Y(1))\|_{L_2(\PP_{X\given T=1})}/\sqrt{n}\big),
\$
where by the given conditions, we have 
\$
\|r(X)\hat{H}(X,Y(1))-H(X,Y(1))\|_{L_2(\PP_{X\given T=1})} 
= O\big(\|\hat\theta - \theta^*\|_{L_2(\PP_{X\given T=1})}\big) = o_P(1).
\$
As a result, we have 
\$
\frac{1}{|\cI_1|} \sum_{i \in \cI_1 }  
r(X_i) \big(\hat{H}(X_i,Y_i) -  H(X_i,Y_i)\big) \leq o_P(1/\sqrt{n}) + 
M \|  \hat\theta - \theta  \|_{L_2(\PP_{X\given T=0})}^2 = o_P(1/\sqrt{n}).
\$
Putting these pieces together, we conclude 
the proof of $\hat\mu_{1,0}^{(j)} = \hat\mu_{1,0}^{*,(j)}+o_P(1/\sqrt{n})$ 
for each $j$. Therefore, averaging over the three folds, we have 
\$
\sqrt{n}\big( \hat\mu_{1,0}^- - \mu_{1,0}^-\big) = \frac{\sqrt{n}}{n_1} \sum_{T_i=1}
\Big(r_{1,0} (X_i) \big( {H}(X_i,Y_i) -  {h}(X_i) \big)- \mu_{1,0}^-\Big) + \frac{\sqrt{n}}{n_0} \sum_{T_i=0} {h}(X_i) + o_P(1/\sqrt{n}),
\$
which, by CLT and Slutsky's theorem, converges in distribution to $N(0,\sigma^2)$. Here $n_1$ is the total number of treated samples, 
and $n_0$ is the number of control samples. The asymptotic variance is 
\$
\sigma^2 = \frac{1}{p_1} \Var\Big( r_{1,0} (X ) \big( {H}(X,Y(1)) -  {h}(X ) \big) \biggiven T=1 \Big) + \frac{1}{p_0} \Var\big( h(X)\given T=0 \big).
\$
where $p_1 = \PP(T=1)$, $p_0 =\PP(T=0)$ and all the 
expectations (variances) are induced by the observed distribution. 

It now remains to show that $\hat\sigma^2\to \sigma^2$ 
in the definition of Theorem~\ref{thm:cond}. 
As $\hat p_1\stackrel{p}{\to} p_1$, $\hat p_0 \stackrel{p}{\to} p_0$, 
by the law of large numbers, 
it suffices to show that 
$\frac{1}{n_1} \sum_{i\in \cI_1}\big(  d_{1,i}^2 - (d_{1,i}^*)^2 \big) = o_P(1)$
and
$\frac{1}{n_1} \sum_{i\in \cI_1} \big(d_{1,i}  -  d_{1,i}^* \big) = o_P(1)$
and similar for $(d_{0,i},d_{0,i}^*)$, where we define the oracle 
counterparts 
\$ 
d_{1,i}^*= r_{1,0}(X_i)\big(H(X_i,Y_i) - h(X_i)\big),\quad 
d_{0,i}^* =  h(X_i). 
\$
By Cauchy-Schwarz inequality, 
\$
\frac{1}{n_1} \sum_{i\in \cI_1}\big(  d_{1,i}^2 - (d_{1,i}^*)^2 \big) 
&= \frac{1}{n_1} \sum_{i\in \cI_1}\big(  d_{1,i} -  d_{1,i}^* \big)^2 
+ \frac{1}{n_1} \sum_{i\in \cI_1} 2\big(  d_{1,i} -  d_{1,i}^* \big) d_{1,i}^* \\
&\leq  \frac{1}{n_1} \sum_{i\in \cI_1}\big(  d_{1,i} -  d_{1,i}^* \big)^2 
+ \sqrt{ \frac{1}{n_1} \sum_{i\in \cI_1}\big(  d_{1,i} -  d_{1,i}^* \big)^2} \cdot  \sqrt{ \frac{1}{n_1} \sum_{i\in \cI_1}\big( d_{1,i}^* \big)^2 }.
\$
Focusing on the summatin within $\cI_1^{(j)}$, 
we have 
\$
\frac{1}{|\cI_1^{(j)}|} \sum_{i\in \cI_1^{(j)}}\big(  d_{1,i} -  d_{1,i}^* \big)^2 
= O_P\Big( \big\| \hat{r}^{(j)}(\hat{H}^{(j)} - \hat{h}^{(j)}) - r_{1,0}(H-h)\big\|_{L_2(\PP_{\cdot\given T=1})}^2 \Big),
\$
where the right-handed side is $o_P(1)$ under the conditions of Theorem~\ref{thm:cond}. 
Other folds and other summation terms follow similar arguments hence $\hat\sigma^2\stackrel{P}{\to} \sigma^2$. 
By Slutsky's lemma, 
we conclude the proof of Theorem~\ref{thm:cond}. 
\end{proof}

\subsection{Proof of Theorem~\ref{thm:cond_robust}}
\label{app:subsec_cond_robust}

\begin{proof}[Proof of Theorem~\ref{thm:cond_robust}] 
The proof follows exactly the same arguments as the proof of Theorem~\ref{thm:cond} with 
$\theta^\diamond$ in place of $\theta^*$, 
where all the errors are controlled in the same way; 
the only difference is to show that 
\$
\EE\big[r(X)\nabla_\theta \ell\big(\theta^\diamond(X),X,Y(1)\big) [\hat\theta(X) - \theta^\diamond(X) ] \biggiven T=1\big]
\$
in parallel with~\eqref{eq:zero_grad} in the proof of Theorem~\ref{thm:cond}. 
We note that this is directly implied by our local condition. 
Therefore, 
under the conditions of Theorem~\ref{thm:cond_robust}, 
\$
\hat\mu_{1,0}^-  = - \frac{1}{n_1}\sum_{i\in \cI_1} r_{1,0}(X_i)\big[  {H}^\diamond(X_i,Y_i(1)) -  {h}^\diamond(X_i ) \big] - \frac{1}{n_0} \sum_{i\in \cI_0} h^\diamond(X_i).
\$
The terms in the summation has expectation 
\$
\mu_{1,0}^\diamond := - \EE\big[r_{1,0}(X_i)   {H}^\diamond(X_i,Y_i(1))  \biggiven T=1\big].
\$
Since $\alpha^*(x),\eta^*(x)$ is the per-$x$ minimizer of $\EE[\ell(\theta,X,Y(1))\given X=x,T=1]$, we have 
\$
\EE\big[   {H}^\diamond(X_i,Y_i(1)) \biggiven X=x,T=1\big] \geq 
\EE\big[   {H} (X_i,Y_i(1)) \biggiven X=x,T=1\big]
\$
for $\PP_{X\given T=1}$-almost all $x$, hence by tower property, we have 
$\mu_{1,0}^\diamond \leq \mu_{1,0}^-$. 
On the other hand, the consistency of $\hat\sigma^2$ to $\Var(\phi_{1,-}^{\diamond}(X,Y,T))$ 
also follows the same arguments as the proof of Theorem~\ref{thm:cond} with 
$\theta^\diamond$, which concludes our proof of Theorem~\ref{thm:cond_robust}. 
\end{proof}

\section{Technical lemmas}

\begin{lemma}\label{lem:cond_to_op}
Let $\cF_n$ be a sequence of $\sigma$-algebra, and let $A_n\geq 0$ be a sequence of nonnegative random variables. If $\EE[A_n\given \cF_n] = o_P(1)$, then $A_n=o_P(1)$. 
\end{lemma}

\begin{proof}[Proof of Lemma~\ref{lem:cond_to_op}]
By Markov's inequality, for any $\epsilon>0$, we have 
\$
B_n := \PP( A_n >\epsilon \given \cF_n) \leq \frac{\EE[A_n\given \cF_n]}{\epsilon } = o_P(1),
\$
and $B_n\in[0,1]$ are bounded random variables. For any subsequence $\{n_k\}_{k\geq 1}$ of $\mathbb{N}$, since $B_{n_k}\pto 0$, there exists a subsequence $\{n_{k_i}\}_{i\geq 1} \subset \{n_k\}_{k\geq 1}$ such that $B_{n_{k_i}} \stackrel{\text{a.s.}}{\to} 0$ as $i\to \infty$. By the dominated convergence theorem, we have $\EE[B_{n_{k_i}}]\to 0$, or equivalently, 
$
\PP(A_{n_{k_i}} >\epsilon) \to 0.
$
Therefore, for any subsequence $\{n_k\}_{k\geq 1}$ of $\mathbb{N}$, there exists a subsequence $\{n_{k_i}\}_{i\geq 1} \subset \{n_k\}_{k\geq 1}$ such that $A_{n_{k_i}} \pto 0$ as $i\to \infty$. By the arbitrariness of $\{n_k\}_{k\geq 1}$, we know $A_n\pto 0$ as $n\to \infty$, which completes the proof. 
\end{proof}